\relax
\documentclass[letterpaper]{article}
\usepackage[preprint]{formatstyle}
\usepackage{times}
\usepackage{helvet}
\usepackage{courier}
\usepackage{url}
\usepackage{graphicx}
\usepackage{amssymb}
\usepackage{amsthm}
\usepackage{amsfonts}
\usepackage{amsmath}
\usepackage{bbm}
\usepackage{bm} 

\newtheorem{lemma}{Lemma}
\newtheorem{proposition}{Proposition}

\newtheorem{innercustomthm}{Proposition}
\newenvironment{customprop}[1]
  {\renewcommand\theinnercustomthm{#1}\innercustomthm}
  {\endinnercustomthm}

\providecommand{\mathbold}[1]{\bm{#1}} 
\newcommand{\R}{\mathbb{R}}

\newcommand{\vct}[1]{\bm{#1}}
\newcommand{\mtx}[1]{\mathbold{#1}}

\DeclareMathOperator*{\argmax}{argmax}
\DeclareMathOperator*{\argmin}{argmin}
\DeclareMathOperator{\conv}{conv}

\title{Robust Market Equilibria with Uncertain Preferences}

\author{
Riley Murray,\textsuperscript{\rm{1,3}}
Christian Kroer,\textsuperscript{\rm{1,4}}
Alex Peysakhovich,\textsuperscript{\rm 2}
Parikshit Shah\textsuperscript{\rm 1} \\
\textsuperscript{\rm 1}Facebook Core Data Science,
\textsuperscript{\rm 2}Facebook Artificial Intelligence Research. \\
\textsuperscript{\rm 3}California Institute of Technology,
\textsuperscript{\rm 4}Columbia University.\\
rmurray@caltech.edu, christian.kroer@columbia.edu, alex.peys@gmail.com, parikshit@fb.com
}

\begin{document}
%

\pagestyle{plain}

\maketitle

\begin{abstract}
The problem of allocating scarce items to individuals is an important practical question in market design. An increasingly popular set of mechanisms for this task uses the concept of market equilibrium: individuals report their preferences, have a budget of real or fake currency, and a set of prices for items and allocations is computed that sets demand equal to supply. An important real world issue with such mechanisms is that individual valuations are often only imperfectly known. In this paper, we show how concepts from classical market equilibrium can be extended to reflect such uncertainty. We show that in linear, divisible Fisher markets a robust market equilibrium (RME) always exists; this also holds in settings where buyers may retain unspent money. We provide theoretical analysis of the allocative properties of RME in terms of envy and regret. Though RME are hard to compute for general uncertainty sets, we consider some natural and tractable uncertainty sets which lead to well behaved formulations of the problem that can be solved via modern convex programming methods. Finally, we show that very mild uncertainty about valuations can cause RME allocations to outperform those which take estimates as having no underlying uncertainty.
\end{abstract}

\section{Introduction}
A key problem in market design is `who gets what'~\citep{roth2015gets}. An important mechanism for the allocation of scarce items to multiple individuals is the use of market equilibrium. In these mechanisms individuals have preferences over scarce items and budgets of money. Prices are set for items such that demand of individuals equals the supply of items. A major issue in practice with such mechanisms is that the utility of an individual is often not known exactly (sometimes even to the individuals themselves). In this paper, we take up the question of computing a robust market equilibrium which takes this imperfect information into account. 

Market equilibrium-based allocations are increasingly prevalent in real world mechanisms and robustness is an important issue in many of these applications. 

In the allocation of courses to students at business schools students report preferences over courses, are given a budget of `fake' currency, and are allocated the courses they receive in the market equilibrium which matches supply of courses to demand \citep{budish2012multi}. Here students may not be fully aware of their own exact valuations for courses. 

In online advertising advertisers report valuations for various impressions, enter budgets, and an auction mechanism again sets supply equal to demand - finding an equilibrium of the market~\cite{balseiro2015repeated,balseiro2017budget,BalseiroGur}. In practice advertisers' bids in the auction (which are equivalent to a paced version of their valuation in equilibrium~\citep{borgs2007dynamics,conitzer2018multiplicative,conitzer2019pacing}) come from a combination of a per-interaction valuation and a machine learning model which predicts the probability that a given user will interact with that ad. For example, an advertiser might say they are willing to pay $\$1$ per click and a model could predict that a user will click a particular ad with probability $.1$-- this would lead to a bid of $10$ cents on the ad when combined with the valuations. Here, again, the valuation is imperfectly known as both the advertiser's willingness to pay and, more so, the click prediction model, come with noise.

We begin with the simplest workhorse model in this literature: linear Fisher markets. In such markets there are buyers and items. Buyers have finite budgets. Items are scarce and divisible. Buyers have utility for each item and value a bundle as the utility-weighted sum of the items in the bundle. We consider both standard Fisher markets where buyer budgets do not have any use outside the market, and quasi-Fisher markets where leftover money has a fixed value for each buyer. 
Equilibria in such markets include allocations to individuals and prices for items, so that buyer demands sum up to supply and allocations satisfy buyer demands. 

In this article, we show how to extend market equilibrium to the notion of robust market equilibrium (RME). Here buyers do not have point-valued utilities for each item, but rather have `uncertainty sets' of possible valuations. In this case a buyer wants to purchase bundles that optimize their worst case utility given prices. This can reflect risk aversion on the part of the buyers (buyers are actually not perfectly aware of item utility) or ignorance on the part of the market designer if the market equilibrium is simply used as an allocative mechanism (e.g. CEEI).
RME can be partly be viewed as a model of classical equilibrium with a particular class of nonlinear utility functions. Indeed, the \textit{existence} of an RME is no different than the existence of a classical equilibrium with appropriate utilities; the main point of departure is how the resulting equilibria are measured. We extend classical notions of envy and regret to reflect uncertainty in buyer valuations, and we give bounds on these quantities as functions of the size of the uncertainty sets. 

It is well-known that computing market equilibria can be computationally intractable for certain market models.
One of the earliest positive results in this area is the celebrated Eisenberg-Gale convex program for Fisher markets ~\citep{eisenberg1959consensus}; this convex program was first proposed for linear utility functions, but was later shown to work for any utilities which are homogeneous of degree one \cite{eisenberg1961aggregation}. In the past 20 years, this convex programming approach has been extended to handle more general Fisher markets such as spending constraint utilities~\citep{birnbaum2011distributed}, utility-capped buyers, quasi-linear buyers~\citep{chen2007note,cole2017convex}, and Fisher markets with transaction costs~\citep{chakraborty2010market}.

Over the course of this article, we show how techniques from the robust optimization literature can be combined with the Eisenberg-Gale convex program to compute RME in both Fisher and quasi-Fisher markets. In principle, this means the existence of polynomial-time algorithms for computing Fisher or quasi-Fisher RME reduces to the existence of a polynomial-time algorithm for evaluating robust utilities. In practice, we need to represent the value of a robust utility as a tractable conic optimization problem. We propose two types of uncertainty sets that lead to convex programs which can be efficiently solved using modern solvers, and we apply these proposals to study RME in real datasets.

To the best of our knowledge, we are the first to study robust market equilibria. However, the topic of robust variants of optimization problems has been studied extensively~\citep{ben2002robust,bertsimas2004price,ben2009robust,bertsimas2011theory}. There, some \emph{nominal} mathematical program is given, and then the robust variant of the program requires that the constraints hold for every instantiation of parameters from some uncertainty set. Our robust market equilibrium approach can be viewed as the natural uncertainty parameterization of Eisenberg-Gale style convex programs. A similar uncertainty-parametrization of utilities was considered by \citet{aghassi2006robust} in the context of game-theoretic equilibrium where they provide a robust analogue of Bayesian equilibrium. Robust game-theoretic equilibria have also been considered in the context of counterfactual prediction~\citep{peysakhovich2019robust}. Finally, there is literature on \emph{robust mechanism design} ~\citep{bergemann2005robust,lopomo2018uncertainty,albert2017mechanism}, where the goal is to design mechanisms that are robust either to uncertainty about the distribution over agent payoffs, or the belief that an agent holds about the types of other agents. Due to the close relationship between mechanism design and market equilibria it would be interesting to understand what kind of robustness properties RME has with respect to traditional robust mechanism design objectives.

\section{Fisher Markets}\label{sec:market_defs}
Consider a market where $n$ \textbf{buyers} compete for $m$ divisible \textbf{goods}, each of unit supply.
Each buyer brings a \textbf{budget} $b_i$ of currency to this market, and a \textbf{utility function} $u_i$ over allocations of the $m$ goods.
The problem of market equilibria is to determine prices for goods and allocations of goods to buyers which satisfy both supply constraints, and certain optimality conditions for each buyer.
The supply constraints do not depend on further details of the market.
If $\mtx{X}$ is a nonnegative matrix where $x_{ij}$ is the \textbf{allocation} of good $j$ to buyer $i$, the allocation is \textbf{feasible} if $\sum_{i=1}^n x_{ij} \leq 1$ for all goods $j$.

The precise form of the buyer optimality conditions depends on whether or not money has value outside the market.
Suppose that the market has assigned each good $j$ a \textbf{price} $p_j$, and assemble these prices into a vector $\vct{p}$ in $\R^m_+$.
When money has no intrinsic value, we are in a \textbf{Fisher market}. A \textbf{market equilibrium} is a feasible allocation $\mtx{X}$ and a set of prices $\vct{p}$ such that each buyer prefers what they are allocated over anything else they could afford:
\begin{equation}
\vct{x}_i \in \argmax\{ u_i(\vct{z}) \,\mid\, \vct{z} \in \R^m_+ \text{ and } \vct{p}^* \cdot \vct{z} \leq b_i \}. \label{eq:Fisher_demandset}
\end{equation}
When money has intrinsic value, buyer utilities become $(\vct{z}, s) \mapsto u_i(\vct{z}) + s$, where $s$ is the buyer's retained budget. Here the buyer optimality conditions for market equilibria are that
\begin{align}
(\vct{x}_i, r_i) \in \argmax\{ u_i(\vct{z}) + s \,\mid\,
    & \vct{z} \in \R^m_{+}, ~ s \in \R_+, \label{eq:quasiFisher_demandset} \\
    & \text{and } \vct{p} \cdot \vct{z} + s \leq b_i \} \nonumber
\end{align}
for the specific value $r_i = b_i - \vct{p} \cdot \vct{x}_i$.
Throughout this article we call this model a \textbf{quasi-Fisher market}, because the effective utility functions $(\vct{z},s) \mapsto  u_i(\vct{z}) + s$ are broadly known as \textit{quasilinear utilities} in the economics literature.

Linear utility functions offer the simplest model for buyer preferences in market equilibrium problems.
In a linear model, each buyer possesses a valuation vector $\vct{v}_i \in \R^m_+$, and assigns utility $u_i(\vct{z}) = \vct{v}_i \cdot \vct{z}$ to a bundle of goods $\vct{z}$.
There is one serious drawback to using linear utilities: in large markets, it is not realistic to assume that all $\vct{v}_i$'s are known exactly-- either by buyers, or by a market maker.

We thus propose a model where each buyer has an associated \textbf{uncertainty set of valuations} $\mathcal{V}_i \subset \R^m_{+}$;
an uncertainty set gives rise to a \textbf{robust utility}
\begin{equation}
u_i(\vct{z}) = \min\{ \vct{z} \cdot \vct{v} \,\mid\, \vct{v} \in \mathcal{V}_i \}. \label{eq:utils}
\end{equation}
It is easy to see that for all nonempty $\mathcal{V}_i$, these robust utilities are concave, and positively homogeneous of degree 1.
We assume these uncertainty sets are nonempty, compact, and satisfy $u_i(\vct{1}) > 0$. With regards to $\mathcal{V}_i$ being nonempty, we will often find it useful to suppose that each buyer has a distinguished \textbf{nominal valuation} $\vct{\hat{v}}_i \in \mathcal{V}_i$.
We often call the function $\vct{z} \mapsto \vct{\hat{v}}_i \cdot \vct{z}$ the \textbf{nominal utility} of buyer $i$.

\subsection{Computing market equilibria}\label{subsec:computing_equilibria}
In a foundational result, \citet{eisenberg1959consensus} showed that equilibria for Fisher markets with linear utility functions can be computed by solving a particular centralized convex program.
In a follow-up work \citet{eisenberg1961aggregation} showed the same convex program can be used to compute Fisher-market equilibria whenever utilities are concave, and homogeneous of degree one.
Much later, \citet{chen2007note} gave a simpler proof for Eisenberg's 1961 result, which extended naturally to compute equilibria in quasi-Fisher markets.\footnote{In situ, Chen et al. refer to quasi-Fisher markets as \textit{mixed Fisher Arrow-Debreu markets}.}

Chen et al.'s result for quasi-Fisher markets requires a very minor modification to the convex program originally proposed by Eisenberg and Gale.
We can state both convex programs by considering a parametric optimization problem: let $\mathrm{Q} \in \{0, 1\}$ be a parameter, where $\mathrm{Q} = 1$ indicates a quasi-Fisher market model, and $\mathrm{Q} = 0$ indicates a Fisher market model.
The \textbf{Eisenberg-Gale convex program} is
\begin{align}
\max &~ \textstyle\sum_{i=1}^n b_i \log(t_i) - \mathrm{Q}\,r_i \label{eq:EG} \\
\text{ s.t.}
    & ~~ (\vct{x}_i, t_i, r_i) \in \mathbb{R}^{m+2}_+  &&\forall~ i \in [n], \nonumber \\
    & ~~ t_i \leq u_i(\vct{x}_i) + \mathrm{Q}\, r_i           &&\forall~ i \in [n], \nonumber \\
    & ~~ \textstyle\sum_{i=1}^n x_{ij} \leq 1   &&\forall~ j \in [m]. \nonumber
\end{align}
The results of Eisenberg and Chen et al. ensure that as long as $u_i$ are concave and homogeneous of degree one, then for $\vct{p} \in \R^m_+$ as the vector of dual variables to the capacity constraints, the matrix $\mtx{X}$ whose rows are $\vct{x}_i$ forms an equilibrium allocation with respect to prices $\vct{p}$.

Because we defined our robust utilities as the minimum over a set of linear functions, we get that robust utilities are homogeneous with degree $1$ and concave. This establishes the following proposition:

\begin{proposition} \label{prop:Fisher}
A solution to \eqref{eq:EG} with $\mathrm{Q} = 0$ produces an equilibrium for the Fisher market where $u_i$ are given by \eqref{eq:utils}.
Optimal prices $\vct{p}$ and allocations $\vct{X}$ satisfy the following properties:
\begin{enumerate}
    \item Each $\vct{x}_i$ is in the demand set of buyer $i$, i.e. \eqref{eq:Fisher_demandset} holds.
    \item Every item $j$ with $p_j >0$ clears the market: $\sum_i \vct{x}_{ij} = 1$.
    \item For every buyer $i$, there exists a $\vct{v}_i^* \in \argmin_{\vct{v} \in \mathcal{V}_i} \vct{v} \cdot \vct{x}_i$ for which allocated items have the same bang per buck:
    \[
    \text{if } ~~ x_{ij}, x_{ik}>0 ~ \text{ then } ~ v_{ij}^{*} / p_j = v_{ik}^{*} / p_k.
    \]
\end{enumerate}
\end{proposition}

The first property in Proposition~\ref{prop:Fisher} is about individual optimality; every buyer seeking to optimize her own utility in the robust sense will find the allocations in her demand set. The second property shows market clearing. The third property shows that the worst case utilities are in fact attained in the allocation, and that an ``equal rates'' condition holds at one of these worst-case utility vectors.

\begin{proposition} \label{prop:quasi_Fisher}
A solution to \eqref{eq:EG} with $\mathrm{Q}=1$ produces an equilibrium for the quasi-Fisher market with $u_i$ given by \eqref{eq:utils}.
Analogous conditions to those in Proposition~\ref{prop:Fisher} hold for prices $\vct{p}$ and allocations $\mtx{X}$.
Furthermore, if $\vct{\beta}$ denotes optimal dual variables to the utility hypograph constraints, then for all buyers $i$
\begin{enumerate}
    \item $\beta_i \leq 1$, and $\beta_i = 1$ whenever $\vct{p} \cdot \vct{x}_i < b_i$.
    \item There is a $\vct{v}_i^* \in \argmin_{\vct{v} \in \mathcal{V}_i} \vct{v} \cdot \vct{x}_i$ so that $\beta_i \leq p_j / v_{ij}^*$ for all $j$, and $\beta_i = p_j / v_{ij}^*$ whenever $x_{ij} > 0$.
  \end{enumerate}
\end{proposition}

The additional conditions concern interpretation of the dual variables $\beta_i$ as \textit{pacing multipliers} in a first-price auction~\citep{conitzer2019pacing}.
The allocation mechanism may be viewed as a first price auction where buyer $i$ competes for item $j$ with a modified or ``paced'' bid of $\beta_i v_{ij}$.
The pacing multipliers are no greater than one (ensuring that no buyer buys a good when its price exceeds its value), there is no unnecessary pacing, and the robust bang-per-buck is equal for all allocated items to buyer $i$.

Propositions~\ref{prop:Fisher} and~\ref{prop:quasi_Fisher} argue that a solution to \eqref{eq:EG} leads to a reasonable solution concept: for risk averse agents that are seeking to robustly maximize their utilities in the face of the market uncertainty, it produces allocations in their individual demand sets. Moreover, it also leads to market clearing, which is desirable for the market designer. Finally, the solution is intuitive, and in fact corresponds to a ``standard'' market equilibrium with respect to a set of attained (worst case) realizations of the utilities. 

From a theoretical perspective, computing market equilibria with robust utilities is no harder than solving an appropriate convex program. Unfortunately, not all convex programs are tractable;
the potential stumbling block in our case is the need to represent the robust utilities via hypographs 
\[
 H_i \doteq \{ (t,\vct{z}) \,\mid\, t \leq u_i(\vct{z}) \}
\]
using a limited library of convex constraints on $t$ and $\vct{z}$.
When the uncertainty sets $\mathcal{V}_i$ have an explicit convex description, the $H_i$ can be represented by appealing to convex duality.
For example, if $\mathcal{V}_i = \{ \vct{v} \,\mid\, \| \vct{v} - \vct{\hat{v}}_i\| \leq \epsilon\}$ for the nominal valuation $\vct{\hat{v}}_i$ and some reference norm $\| \cdot \|$, the hypograph can be represented as
\[
H_i = \{ (t, \vct{z}) \,\mid\, t + \epsilon \|\vct{z}\|^{*} \leq \vct{\hat{v}}_i \cdot \vct{z} \}
\]
where $\| \cdot \|^*$ is the dual norm to $\|  \cdot \|$.
In cases such as this, we can appeal to standard convex programming solvers, which allow only specific convex operations.

In principle, the $\mathcal{V}_i$ need not be convex. This is for the following reason: if ``$\conv$'' is the operator that computes a set's convex hull, then
\[
\min\{ \vct{z} \cdot \vct{v} \,\mid\, \vct{v} \in \mathcal{V}_i \} = \min\{ \vct{z} \cdot \vct{v} \,\mid\, \vct{v} \in \conv \mathcal{V}_i \}
\]
for all $\vct{z}$ in $\R^m$.
Thus we may replace nonconvex uncertainty sets by their convex hulls without loss of generality. 
The challenge of nonconvex uncertainty sets then reduces to the problem of efficiently representing their convex hulls.

\subsection{A dual formulation}
The literature on theoretical analysis of market equilibrium often benefits from analyzing primal and dual formulations in conjunction with one another.
Given the nature of our utility functions, it is reasonably straightforward to compute the dual by appealing to Fenchel duality.
The drawback to working with a Fenchel dual is that it can become harder to interpret dual variables.
To assist others in future theoretical analysis of market equilibria with robust utility functions, we provide the following result (the proof of which is in Appendix B, along with all other omitted proofs of claims in this article):

\begin{proposition}
Let $u_i$ by given by \eqref{eq:utils} where $\mathcal{V}_i \subset \R^m_+$  are nonempty compact convex sets, and $u_i(\vct{1}) > 0$.
It can be shown that the following problem is a dual to \eqref{eq:EG} 
\begin{align}
\min ~& \vct{p} \cdot \vct{1} - \sum_{i=1}^n\left\{ b_i + b_i \log(\beta_i/b_i) \right\} \label{eq:EG_dual} \\
    \text{s.t. }& \vct{p} \in \R^m_{+},~ \vct{\beta} \in \R^n_+ \nonumber \\
                & \mathrm{Q}[1-\vct{\beta}] \geq \vct{0} \nonumber \\
                & \vct{p} \geq \beta_i \vct{v}_i \qquad ~ \forall ~ i \in [n] \nonumber \\
                & \vct{v}_i \in \mathcal{V}_i ~~~ \qquad ~ \forall ~ i \in [n]. \nonumber
\end{align}
\end{proposition}

The formulation is nonconvex as stated, due to bilinear inequality constraints $\vct{p} \geq \beta_i \vct{v}_i$ in the variables $\beta_i$ and $\vct{v}_i$. This bilinearity can be represented in a convex way by using perspective functions.
For each $i \in [n]$, set $\mathcal{W}_i = \R^m_+ + \mathcal{V}_i$.\footnote{Addition here denotes the Minkowski sum.} The $\mathcal{W}_i$ are convex sets, and so their indicator functions
\[
\delta_{\mathcal{W}_i}(\vct{y}) = \begin{cases} 0  &\text{ if } \vct{y} \in \mathcal{W}_i \\
+\infty &\text{ otherwise}
\end{cases}
\]
are also convex. The bilinear constraints can thus be written
\[
 \beta_i \delta_{\mathcal{W}_i}(\vct{p} / \beta_i) \leq 0 ~ \forall ~ i \in [n].
\]
These constraints are convex, since the mappings on the left-hand-side of the inequalities are the perspectives of the convex functions $\delta_{\mathcal{W}_i}$.
It can be shown that this formulation (in terms of perspectives-of-indicators) is in fact precisely what results when applying Fenchel duality to \eqref{eq:EG} directly.

\section{Properties of Robust Equilibria}

On the surface, the market equilibria considered in this article are no different than classical market equilibria with a particular choice of nonlinear utility function. 
Thus there are standard properties of market equilibria which are naturally satisfied in our setting.
For example, if $(\mtx{X},\vct{p})$ form an equilibrium in a Fisher market, and all budgets $b_i = 1$, then the allocations are \textit{envy-free} in the sense that for all buyers $i$
\[
u_i(\vct{x}_i) \geq u_i(\vct{x}_{i'}) \text{ for all } i' \text{ in } [n].
\]
Similarly, the allocations have \textit{no regret}, in that for all $i$
\[
R_i(\vct{y},\vct{p}) \doteq \max\{ u_i(\vct{z}) - u_i(\vct{y}) \,\mid\, \vct{z} \in \R^m_+,~ \vct{p} \cdot \vct{z} \leq b_i\}
\]
satisfies $R_i(\vct{x}_i,\vct{p}) = 0$.
These classical measures of fairness are important, but they do not reflect our reason for choosing the robust utilities $u_i(\vct{z}) = \min\{ \vct{z} \cdot \vct{v} \,\mid\, \vct{v} \in \mathcal{V}_i \}$ in the first place.

We assume that once a buyer has received a bundle of goods $\vct{x}_i$, some $\vct{v}_i \in \mathcal{V}_i$ instantiates as the buyer's true valuation.
This can occur, for example, if uncertainty came from the fact that buyers did not have perfect clarity into what they valued during evaluation time but gain this information when they actually receive the good.
In this model it is natural to evaluate metrics such as envy, regret, or Nash Social Welfare with respect to these \textbf{realizations} of buyer uncertainty.
We focus on the scenarios when the valuation realizes to its nominal value, and when it realizes adversarially.

\subsection{Adversarially-robust envy}
Let $\mtx{X}$ denote a matrix whose rows $\vct{x}_i$ are equilibrium allocations with respect to prices $\vct{p}$ and robust utilities $u_i$.

Also let $r_i = b_i - \vct{p} \cdot \vct{x}_i$ denote the retained budget for each buyer $i \in [n]$ (understanding that $r_i = 0$ in Fisher markets).
In these terms, robust envy is
\begin{align*}
\mathcal{E}_i(\mtx{X},\vct{p}) = \max & ~(\vct{v} \cdot \vct{x}_{i'} + r_{i'}) - (\vct{v} \cdot \vct{x}_i + r_{i}) \\
    \text{s.t } & ~ \vct{v} \in \mathcal{V}_i,~ i' \in [n].
\end{align*}
In order to analyze robust envy, it is helpful to frame things in terms of the underlying robust utility functions.
One may verify that the following identity holds
\[
\mathcal{E}_i(\mtx{X},\vct{p}) = \max_{i' \in [n]}\{ r_{i'} - r_i - u_i(\vct{x}_i - \vct{x}_{i'}) \}.
\]
It is therefore evident that we can efficiently compute robust envy whenever we can compute the utility functions $u_i$.

\subsection{Adversarially-robust regret}
Given prices $\vct{p}$, the $i^{\text{th}}$ buyer takes an action by drawing from the following polytope
\begin{align*}
Z_i(\vct{p}) = \{ \vct{z} \,\mid\,
    & \vct{0} \leq \vct{z} \leq \vct{1},\,\vct{p} \cdot \vct{z}  \leq b_i \}.
\end{align*}
Then given an allocation $\vct{x}_i$, we thus define a buyer's robust regret as
\begin{align*}
\mathcal{R}_i(\vct{x}_i,\vct{p}) = \max & ~ [\vct{v} - \mathrm{Q} \, \vct{p}] \cdot \vct{z} - [\vct{v} - \mathrm{Q}\, \vct{p}] \cdot \vct{x}_i \\
    \text{s.t. } & ~ \vct{v} \in \mathcal{V}_i, ~ \vct{z} \in Z_i(\vct{p})
\end{align*}
-- where ``$\mathrm{Q}$'' is 1 if the buyer is in a quasi-Fisher market, and 0 if otherwise.
Mirroring our approach to robust-envy, robust-regret can be expressed with the robust utility functions.
Specifically,
\begin{align*}
\mathcal{R}_i(\vct{x}_i,\vct{p}) = \max & ~ \mathrm{Q} \, \vct{p} \cdot [\vct{x}_i - \vct{z}] - u_i(\vct{x}_i - \vct{z}) \\
    \text{s.t. } & ~ \vct{z} \in Z_i(\vct{p}).
\end{align*}

Robust regret is harder to compute than robust envy.
The issue is that the objective appearing in the second expression for $\mathcal{R}_i(\vct{x}_i,\vct{p})$ is convex, rather than concave.
From a complexity perspective, it is very difficult to maximize convex functions.
Our one reason for hope is that we are trying to maximize a convex function over the \textit{polytope} $Z_i(\vct{p})$.
In this specific case we may use the fact that the maximum of a real-valued convex function over a polytope is always attained at one of the polytope's extreme points.
Thus in principle, it is possible to compute $\mathcal{R}_i(\vct{x}_i,\vct{p})$ by evaluating the objective function at each extreme point of $Z_i(\vct{p})$, and taking the largest value.
This approach may be viable in Fisher markets for buyers that are heavily budget constrained, since if $b_i \leq p_j$ for all goods $j$, then the extreme points of $Z_i(\vct{p})$ are $\vct{z}^{(0)} = \vct{0}$, and $\vct{z}^{(j)} = b_i / p_j \vct{e}_j$ for all $j \in [m]$.\footnote{Here, $\vct{e}_j$ denotes the $j^{\text{th}}$ standard basis vector in $\R^m$.}
In other contexts the extreme-point approach will likely be impractical.
For example, in Fisher markets where prices are low compared to a buyer's budget, the set $Z_i(\vct{p})$ can contain as many as $2^m$ extreme points.

\subsection{Bounds in Fisher markets}

In this section we describe elementary ways to bound robust envy and robust regret for Fisher markets.

\begin{lemma}\label{lem:L1_width}
Suppose the $\ell_1$-width of $\mathcal{V}_i$ is bounded above by $w$. Then for all $\vct{v}$ in $\mathcal{V}_i$ and all $\vct{z}$ with $\| \vct{z} \|_{\infty} \leq 1$, we have 
\[
-w + \vct{v} \cdot \vct{z} \leq u_i(\vct{z}) \leq \vct{v} \cdot \vct{z} + w.
\]
\end{lemma}
\begin{proof}
Fix $\vct{v}$ in $\mathcal{V}_i$ and $\vct{z}$ with $\| \vct{z}\|_{\infty} \leq 1$.
Since $\mathcal{V}_i$ is compact, there exists a $\vct{v}^\star \in \mathcal{V}_i$ so that $u_i(\vct{z}) = \vct{v}^\star \cdot \vct{z}$.
Since the $\ell_1$-width of $\mathcal{V}_i$ is bounded by $w$, we have that $\Delta \vct{v} = \vct{v}^\star - \vct{v}$ satisfies $\| \Delta \vct{v} \|_1 \leq w$.
Writing $u_i(\vct{z}) = \vct{v} \cdot \vct{z} + \Delta \vct{v} \cdot \vct{z}$, we can invoke duality between the $\ell_1$ and $\ell_\infty$ norms to see that  $|\Delta \vct{v} \cdot \vct{z}| \leq w$. The result follows.
\end{proof}
Lemma \ref{lem:L1_width} is useful because the natural domain for $u_i$ in robust regret and envy is the unit $\ell_\infty$ ball.

\begin{proposition}
Suppose the $\ell_1$-width of $\mathcal{V}_i$ is bounded above by $w$. Then robust regret is bounded by $2w$. If all budgets are equal then robust is envy is bounded by $2w$.
\end{proposition}
\begin{proof}
Let $\vct{X},\vct{p}$ be a robust market equilibrium. 
Let $i^*,\vct{v}^*$ be the maximizers of robust envy for buyer $i$. We then have
\begin{align*}
    \mathcal{E}_i(\mtx{X},\vct{p})
    &=
    ( \vct{v}^*\cdot \vct{x}_{i^*} + r_{i^*} )
    - ( \vct{v}^*\cdot \vct{x}_{i} + r_{i} ) \\
    &\leq
    ( w + u_i(x_{i^*}) + r_{i^*})
    - ( -w + u_i(x_i) + r_{i} ) \\
    &=
    2w + ( u_i(x_{i^*}) + r_{i^*})
    - (u_i(x_i) + r_{i} )) \\
    &\leq 2w,
\end{align*}
where the first inequality follows by Lemma~\ref{lem:L1_width} (since supplies are $1$),
and the second inequality follows by the no-envy property since $\mtx{X}, \vct{p}$ constitute a market equilibrium with respect to the robust utilities and budgets are equal.

The proof for robust regret is analogous, but using the fact that $u_i(\vct{x}_i) - \mathrm{Q}\, \vct{p} \cdot \vct{x_i} \geq u_i(\vct{z}) - \mathrm{Q}\, \vct{p} \cdot \vct{z}$ where $\vct{z}$ is the allocation for robust regret.
\end{proof}

While the above bounds hold uniformly for every agent $i$, in the specific case when the $\ell_{\infty}$-width of the uncertainty sets are small one can also obtain a bound on the \emph{average} robust regret of all the buyers.

\begin{proposition} \label{prop:avg_regret}
Suppose that for each buyer $i \in [n]$, we have $|v_{ij} - v_{ij}^{\prime}| \leq R$  for all $j \in [m]$ and all $\vct{v}_{i}, \vct{v}_{i}^{\prime} \in \mathcal{V}_i$. Then the following holds:
\[
\frac{1}{n}\sum_{i=1}^n \text{Robust-Regret}_i \leq \frac{2R m}{n}.
\]
\end{proposition}
In a scarce market where the number of buyers far exceeds the number of items, the bound in Proposition \ref{prop:avg_regret} guarantees that the average regret is small.

\section{Concrete Models for Buyer Uncertainty} \label{subsec:models_for_experiments}

We consider the case of low rank markets~\citep{kroer2019computing,peysakhovich2019fair}. In these markets, the valuations individuals place on items are not independent, and can be predicted from one another; this is common in most real-world allocation problems (e.g. valuations for courses are correlated within person, ratings of movies are correlated within person, willingness to pay for different impressions are correlated within advertiser). 

Let $\mtx{\widehat{V}}$ denote our estimated valuation matrix for the Fischer market, i.e.  $\hat{v}_{ij}$ is the value that individual $i$ places on item $j$. The low rank model assumes that  $\mtx{\widehat{V}} = \mtx{\widehat{\Theta}} \mtx{\widehat{\Phi}}^\intercal$ where $\mtx{\widehat{\Theta}}$ and $\mtx{\widehat{\Phi}}$ are embedding matrices for individuals and items respectively, of sizes $n \times d$ and $m \times d$. The problem is typically interesting in the regime where where $d << n, m$. 

In such a model, there are a few natural sources of uncertainty.
For example, in practice, individual and item embeddings would typically be estimated using some matrix completion procedure.
The matrix completion procedure may introduce errors (relative to the ground-truth completed matrix), and these errors induce uncertainty on factors $\mtx{\widehat{\Theta}}$ or $\mtx{\widehat{\Phi}}$.
There is also the possibility that the low-rank assumption does not well reflect reality; this would correspond to uncertainty in small-norm but high-rank perturbations to $\mtx{\widehat{V}}$.

When the embedding model is imperfect and subject to uncertainty, a natural description one may consider is the \textbf{joint uncertainty model} defined by the set $\mathcal{V}^{J}_i(\epsilon_1, \epsilon_2)$ as:

\begin{align*}
\{ &\vct{\theta} \mtx{\Phi}^\intercal \,\mid\,
     \vct{\theta} \in \mathbb{R}^{1 \times d}, ~ \vct{\theta} \mtx{\Phi}^\intercal \geq \vct{0}, \;
     \| \vct{\theta} - \vct{\hat{\theta}}_i \|_p \leq \epsilon_1 \| \vct{\hat{\theta}}_i \|_p , \\
     &\| \mtx{\Phi} - \mtx{\widehat{\Phi}} \|_q \leq \epsilon_2 \| \mtx{\widehat{\Phi}} \|_q , \;
     \vct{\theta} \mtx{\Phi}^\intercal \cdot \vct{1} = \vct{\hat{\theta}}_i \mtx{\widehat{\Phi}}^\intercal \cdot \vct{1}  \}.
\end{align*}
This model captures uncertainty in both the buyer-side and item-side embedding representations. The model parameters include a pair of norms (the $p$-norm for the vector quantity, and the $q$-norm for the matrix quantity),\footnote{The ``$p$'' in ``$p$-norm'' should not be confused with vectors $\vct{p}$ of market prices for goods.} as well as relative radii $\epsilon_1, \epsilon_2 \in (0, 1)$.

Certain special cases of the joint uncertainty model are interesting to consider.
The \textbf{direct model} for uncertainty under such parameters is realized when $\epsilon_2=0$, and $\mtx{\widehat{\Phi}} = I$, the identity matrix. In this situation, we have:
\begin{align*}
\mathcal{V}_i^{d}(p,\epsilon) = \{ \vct{v} \,\mid\,
    & \vct{v} \in \mathbb{R}^{m}, ~ \vct{v} \geq \vct{0}, ~ \\
    & \| \vct{v} - \vct{\hat{v}}_i \|_p \leq \epsilon \| \vct{\hat{v}}_i \|_p \\
    &  \vct{v} \cdot \vct{1} = \vct{\hat{v}}_i \cdot \vct{1}\}.
\end{align*}
When $\epsilon_2 = 0$ but $\mtx{\widehat{\Phi}}$ is arbitrary, we obtain the $\textbf{buyerside model}$, which is described by:
\begin{align*}
\mathcal{V}^{b}_i(p,\epsilon) = \{ \vct{\theta} \mtx{\widehat{\Phi}}^\intercal \,\mid\,
    & \vct{\theta} \in \mathbb{R}^{1 \times d}, ~ \vct{\theta} \mtx{\widehat{\Phi}}^\intercal \geq \vct{0}, \\
    & \| \vct{\theta} - \vct{\hat{\theta}}_i \|_p \leq \epsilon \| \vct{\hat{\theta}}_i \|_p , \\
    & \vct{\theta} \mtx{\widehat{\Phi}}^\intercal \cdot \vct{1} = \vct{\hat{\theta}}_i \mtx{\widehat{\Phi}}^\intercal \cdot \vct{1}  \}.
\end{align*}

The linear equality constraints in these definitions require that we conserve utility about the vector of all ones (``$\vct{1}$''); there are several compelling reasons for doing this.
First, it enforces a kind of regularity condition: assuming $\vct{\hat{v}}_i$ is nonzero, the robust utilities will be positive under the uniform allocation $\vct{x}_i = \vct{1} / m$, regardless of $\epsilon \geq 0$ and $p \geq 1$.
Positivity of robust utilities under the uniform allocation ensures that the convex program \eqref{eq:EG} will always be feasible.
The second reason for conserving utility about $\vct{1}$ is a concentration argument: although it is unrealistic to presume that \textit{all} entries $\hat{v}_{ij}$ are known to high precision, mild statistical assumptions allow us to reliably estimate the sum $\sum_{j=1}^m \hat{v}_{ij} = \vct{\hat{v}}_i \cdot \vct{1}$.
Finally, because all $\vct{v}_i$ in the uncertainty sets are elementwise nonnegative, the conservation constraint provides a scale-invariance: $\| \vct{v}_i\|_1 = \| \vct{\hat{v}}_i\|_1 $.
This prevents us from thinking that some $\vct{v}_i$ in $\mathcal{V}_i$ are more likely than others, simply because they have smaller norm.

We now turn to a deriving a tractable representation for the hypographs of $u_i$ under the direct model.
Fix $i \in [n]$, $p \geq 1$, and $\epsilon \in (0, 1)$.
For $\vct{\lambda}$ in $\R^m_+$ and $\mu$ in $\R$, consider the parameterized Lagrangian
\[
\mathcal{L}(\vct{v},\vct{\lambda},\mu;\vct{z}) =
\vct{v} \cdot \vct{z} + \mu (\vct{v} \cdot \vct{1} - \vct{\hat{v}}_i \cdot \vct{1}) - \vct{v} \cdot \vct{\lambda}
\]
where $\vct{v}$ is restricted to $D_i \doteq \{ \vct{v} \,\mid\, \| \vct{v} - \vct{\hat{v}}_i \|_p \leq \epsilon \| \vct{\hat{v}}_i\|_p\}$.
Note how the Lagrangian is convex in $\vct{v}$, concave in $\vct{\lambda}, \mu$, and how $D_i$ is a compact convex set. We can thus invoke strong duality in a minimax representation for $u_i(\vct{z})$:
\begin{align*}
  u_i(\vct{z})
  &= \min_{\vct{v} \in D_i}\left\{ \max_{\substack{\vct{\lambda} \geq \vct{0} \\ \mu \in \R}} \{ \mathcal{L}(\vct{v},\vct{\lambda},\mu;\vct{z}) \} \right\} \\
  &= \max_{\substack{\vct{\lambda} \geq \vct{0} \\ \mu \in \R}}  \left\{ \min_{\| \Delta \vct{v} \|_p \leq \epsilon \| \vct{\hat{v}}_i \|_p} \{ \mathcal{L}(\vct{\hat{v}}_i + \Delta \vct{v},\vct{\lambda},\mu;\vct{z}) \} \right\} \\
  &= \max_{\substack{\vct{\lambda} \geq \vct{0} \\ \mu \in \R}}\bigg\{ \vct{\hat{v}}_i \cdot [\vct{z} - \vct{\lambda}] - \epsilon \| \vct{\hat{v}}_i \|_p \| \vct{z} - \vct{\lambda} + \mu \vct{1}  \|_p^* \bigg\}
\end{align*}
-- where $\| \cdot \|_p^*$ is dual to the $p$-norm.
The third identity allows us to represent the hypograph of $u_i$ in the direct uncertainty model as
\begin{align*}
H_i = \{ (\vct{z},t) \,\mid\,
    & (\vct{z},t,\mu,\tau,\vct{\lambda}) \in \R^{m+3} \times \R^m_+ \\
    & \| \vct{z} - \vct{\lambda} + \mu \vct{1}  \|_p^* \leq \tau / ( \epsilon \| \vct{\hat{v}}_i \|_p) \\
    & t + \tau \leq \vct{\hat{v}}_i \cdot [\vct{z} - \vct{\lambda}]\}.
\end{align*}

Now we address the buyerside uncertainty model. 
The derivation proceeds in an identical fashion, by minimizing an appropriate Lagrangian over the compact convex set $\{ \vct{\theta} \,\mid\, \| \vct{\theta} - \vct{\hat{\theta}}_i \|_p \leq \epsilon \| \vct{\hat{\theta}}_i \|_p\}$.
The outcome of this process is that $u_i(\vct{z}) = \min\{ \vct{z} \cdot \vct{v} \,\mid\, \vct{v} \in \mathcal{V}_i^{b}(p,\epsilon) \}$ has hypograph
\begin{align*}
H_i = \{ (\vct{z},t) \,\mid\,
    & (\vct{z},t,\mu,\tau,\vct{\lambda}) \in \R^{m+3} \times \R^m_+  \\
    & \| \mtx{\Phi}^\intercal [ \vct{z} - \vct{\lambda} + \mu \vct{1}] \|_p^* \leq \tau / (\epsilon \| \vct{\hat{v}}_i \|_p) \\
    & t + \tau \leq \vct{\hat{v}}_i \cdot [\vct{z} - \vct{\lambda}]\}.
\end{align*}
In the next section, we provide numerical experiments that investigate robust equilibria with respect to direct and buyer-side models.

In contrast to the direct and the buyer-side uncertainty models which have tractable representations, the convex hull of the fully general joint uncertainty model given by $\mathcal{V}^{J}_i(\epsilon_1, \epsilon_2)$ does not seem to be amenable to an exact, tractable representation. In the special case when $p$ corresponds to the Euclidean norm and $q$ corresponds to the Frobenius norm, we provide a tractable representation of an \textit{outer approximation} of this set. Robust equilibria may be computed with respect to these outer approximations; since 
\[
\min_{\vct{v} \in \mathcal{V}^{out}_i(\epsilon_1, \epsilon_2)} \vct{v} \cdot \vct{z} \leq \min_{\vct{v} \in \mathcal{V}^{J}_i(\epsilon_1, \epsilon_2)} \vct{v} \cdot \vct{z}  \qquad \forall ~ \vct{z} ~ \in \R^m,
\]
the resulting allocations will be guaranteed to be robust against all uncertainty realizations in $\mathcal{V}^{J}_i(\epsilon_1, \epsilon_2)$.

\begin{proposition}\label{prop:joint_outer_approx}
Let ${\mathcal{V}}^{out}_i(\epsilon_1, \epsilon_2)$ denote the set
\begin{align*}
\{ \vct{v}  \,:\,
&  \vct{v} = \vct{\mu} + \vct{\delta} \mtx{\widehat{\Phi}}^\intercal + \vct{\hat{\theta}}_i \mtx{\Delta}^\intercal +  \vct{\hat{\theta}}_i \mtx{\widehat{\Phi}}^\intercal, \\
& \vct{v} \geq 0, ~ \vct{v} \cdot \vct{1} = \vct{\hat{v}}_i \cdot \vct{1},~  \| \vct{\delta} \|_2 \leq \epsilon_1 \| \vct{\hat{\theta}}_i \|_2,   \\
& \|\vct{\mu} \|_2 \leq \epsilon_1 \epsilon_2  \| \vct{\hat{\theta}}_i\|_2 \|\mtx{\widehat{\Phi}} \|_F , \; \| \mtx{\Delta} \|_F \leq \epsilon_2 \| \mtx{\widehat{\Phi}} \|_F  \}.
\end{align*}
 Then
 \[
 \conv\left(\mathcal{V}^{J}_i(\epsilon_1, \epsilon_2)\right) \subseteq \mathcal{V}^{out}_i(\epsilon_1, \epsilon_2).
\]
\end{proposition}
Note that $\mathcal{V}^{out}_i(\epsilon_1, \epsilon_2)$ is a convex body, and that it is tractable to optimize over since it only involves second-order cone constraints.

\section{Experimental Results}\label{sec:experimental_results}
We construct low rank markets following \citet{kroer2019computing}.
We start with the MovieLens 1M dataset where $\approx 6000$ individuals give ratings to $\approx 4000$ movies.
We use standard techniques to complete the matrix and take the $200$ individuals with the most movies rated; these individuals are endowed with a unit budget of fictitious currency. 
From this we construct a ``plentiful'' market with $m=500$ goods (movies), and a ``scarce'' market with $m=50$ goods.
We consider these markets in Fisher and quasi-Fisher settings, and with both direct and buyerside models from Section \ref{subsec:models_for_experiments}.
For the buyerside models, $\mtx{\hat{\Phi}}$ are of size $m \times d$, where $d$ is the rank of the nominal valuation matrix; our scarce and plentiful markets had ranks 25 and 35 respectively.
Experiments here use the 2-norm; refer to Appendix A for the same experiments under 1-norm uncertainty.

For our implementation we rely on CVXPY 1.0~\citep{cvxpy,cvxpy_rewriting} to interface with solvers MOSEK~\citep{mosek2010mosek,dahl2019primal} and ECOS~\citep{domahidi2013ecos}. MOSEK is used to solve the equilibrium problems; although other solvers exist which support the logarithmic terms in the objective, it is our experience that no other solver can reliably handle equilibrium problems beyond very small scales. 
EOCS is used to evaluate utility functions as part of computing robust envy and robust-utility Nash welfare (the latter metric we define momentarily).

\subsection{Nash welfare}

Nash Social Welfare (the product of agent's utility functions) is a popular measure of community utility.
Here we adopt a normalized version of Nash Social Welfare: the budget-weighted geometric means of utility functions; we consider this quantity with respect to nominal utilities $\vct{x}_i \mapsto \vct{\hat{v}}_i \cdot \vct{x}_i + r_i$ and robust utilities $\vct{x}_i \mapsto u_i(\vct{x}_i) + r_i$.
By evaluating the nominal-utility Nash welfare at the robust solution, we get a sense of the ``price of robustness.''
By evaluating the robust-utility Nash welfare at the nominal solution, we can measure the price of overconfidence in the point estimate $\{ \vct{\hat{v}}_i \} \approx \mathcal{V}_i$.

In every single experiment we conducted, the nominal Nash welfare of the robust solution decayed slower than the robust Nash welfare of the nominal solution.
This is to say: the potential price of robustness was usually modest, relative to the potential price of optimism.

\begin{center}
\includegraphics[width=0.6\columnwidth]{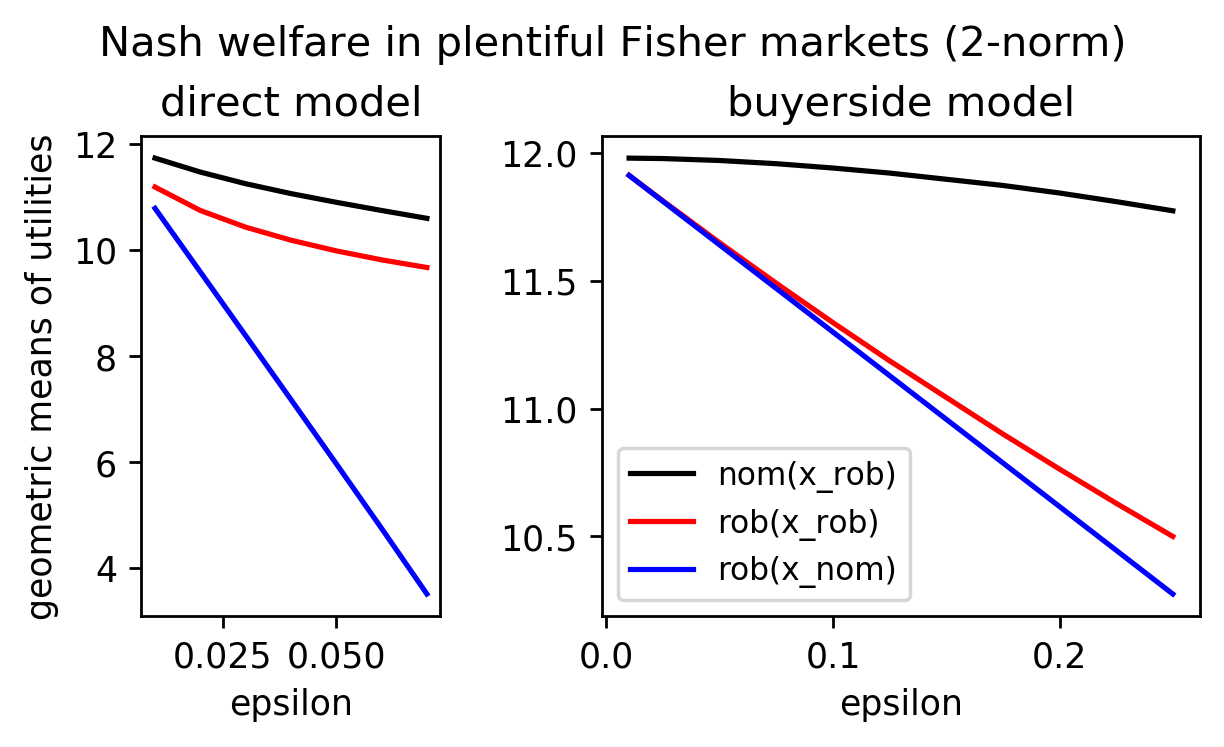}
\end{center}

The robust Nash welfare of a nominal solution is very sensitive to direct uncertainty, while it is relatively stable for buyerside uncertainty.
This trend holds in Fisher and quasi-Fisher markets.
In quasi-Fisher models, the nominal Nash welfare of a robust solution can be \textit{larger} than the nominal Nash welfare of the nominal solution.
This is surprising, since robust solutions are optimizing for a different objective than nominal Nash welfare.
 
\begin{center}
\includegraphics[width=0.6\columnwidth]{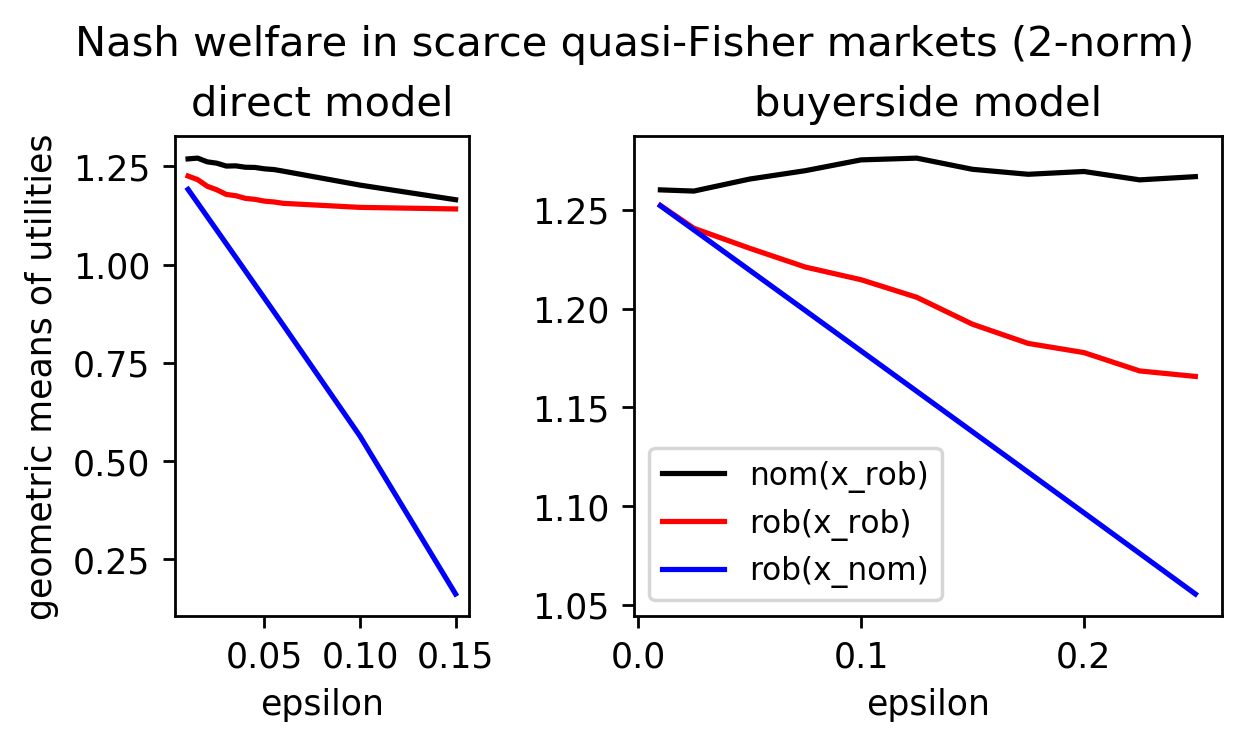}
\end{center}

Qualitatively, this can be attributed to uncertainty causing some buyers to exit the market, which drives down prices for those that remain in the market.
We found that this mostly happens with the buyerside uncertainty model in scarce markets.
It is unclear why this happens more with buyerside than direct uncertainty models, however it is very reasonable that this happens more with scarce rather than plentiful markets.

\subsection{Robust envy}

In the last section we saw how the robust Nash welfare of a nominal solution is very sensitive to direct uncertainty, and is relatively stable for buyerside uncertainty.
These trends also hold for robust envy.

\begin{center}
\includegraphics[width=0.6\columnwidth]{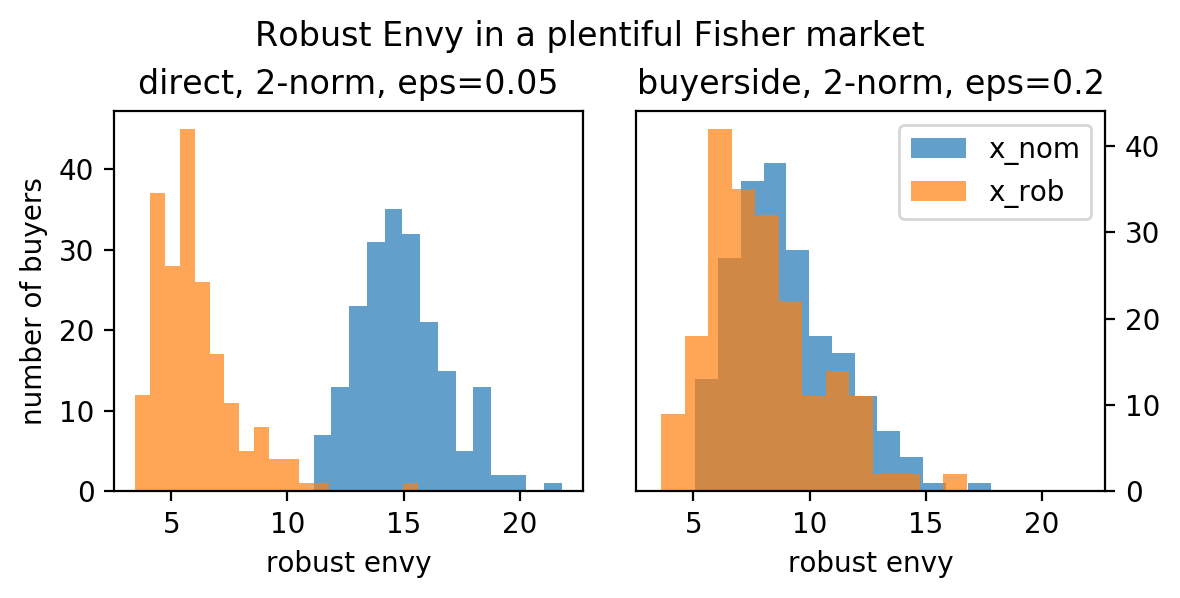}
\end{center}

The benefits of a robust solution persist in scarce markets, although the effects are less pronounced here.
Out of all experiments we conducted, the figure below illustrates the case where robust envy distributions exhibited the most overlap.
Even in this case, there is a clear performance benefit of the robust solution, compared to the nominal solution.

\begin{center}
\includegraphics[width=0.6\columnwidth]{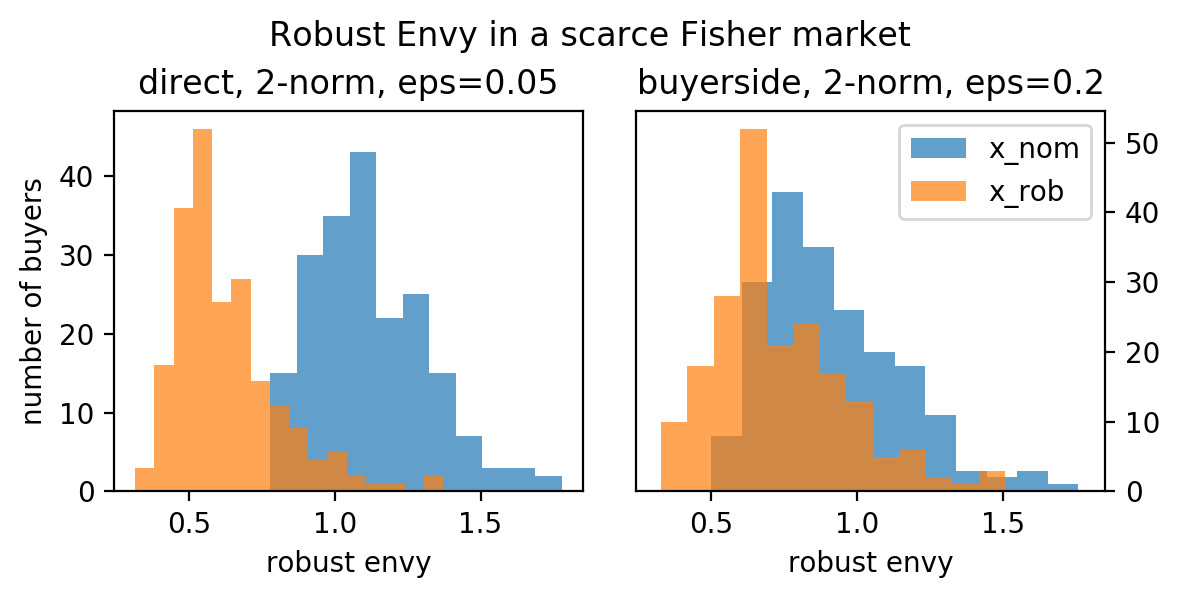}
\end{center}

\subsection{Equilibrium prices in quasi-Fisher markets}

The figure below plots equilibrium prices of every good $j$ as the uncertainty radius $\epsilon$ ranges from 0 to 0.5.
These lines have reduced opacity, so that areas of higher price density can be easily discerned. The minimum, maximum, and mean prices are traced by solid black lines.

\begin{center}
\includegraphics[width=0.6\columnwidth]{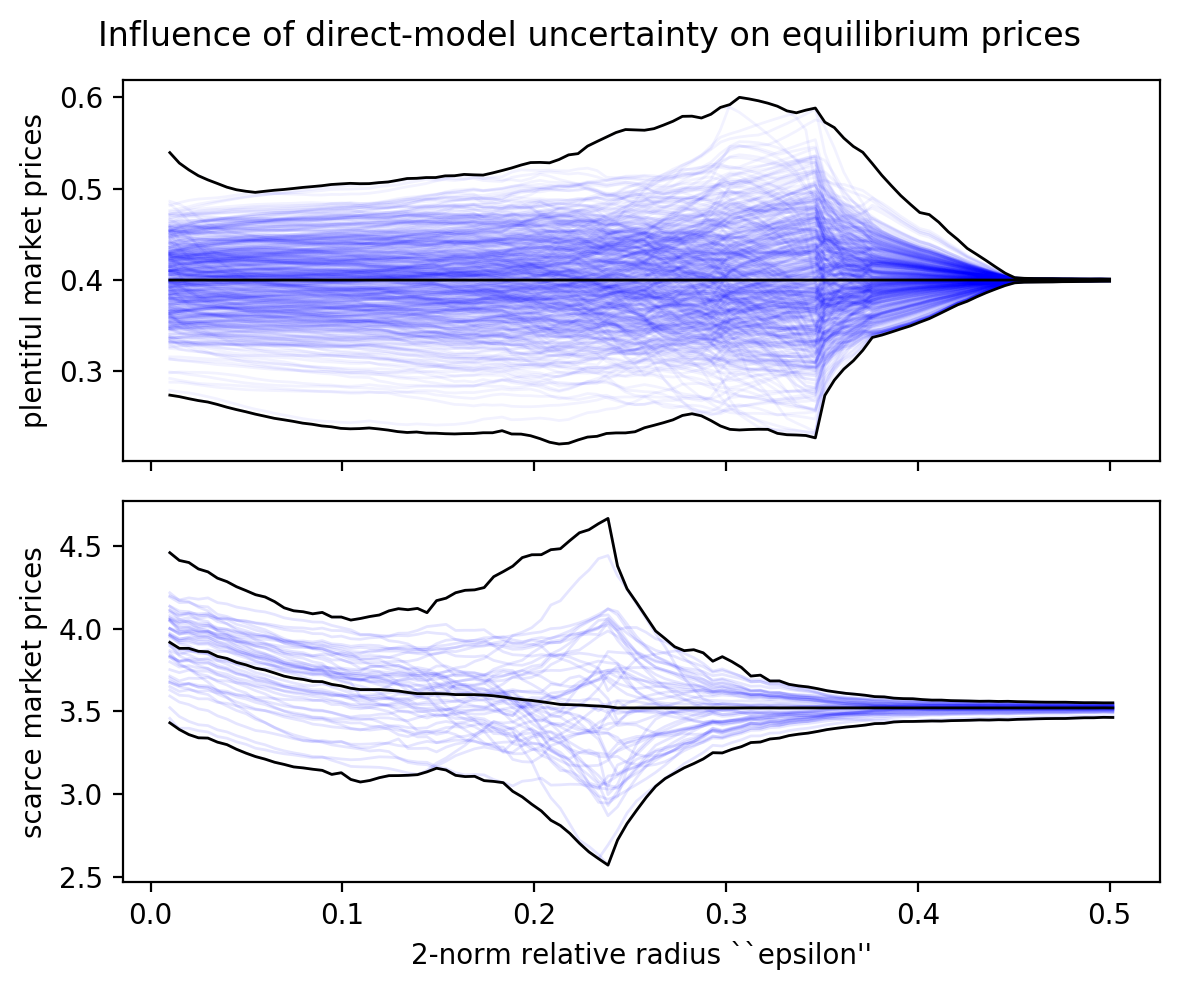}
\end{center}

There are several trends of note in this figure.
First, we see that in the plentiful market, the average price remains constant even as uncertainty increases.
Thus in the plentiful setting, even very large amounts of uncertainty may not dissuade buyers from participating in the market.
In the scarce market we see a different outcome: the average price drops slowly and steadily, from just below 4.0 to just above 3.5.
There are still important commonalities between the plentiful and scarce markets.
The simplest properties are that prices $p_j$ do not evolve monotonically as $\epsilon$ increases, and that uncertainty can cause changes in the order of goods when sorting by $p_j$.
Another crucial property is that as $\epsilon$ gets particularly large, the prices converge to a common value.
Such convergence agrees with our intuition that for large enough $\epsilon$, all buyers' uncertainty sets  will reduce to scaled standard simplices $\mathcal{V}_i^{d}(p,\epsilon) = \{ \vct{v} \,\mid\, \vct{0} \leq \vct{v},~ \vct{v} \cdot \vct{1} = \vct{\hat{v}}_i \cdot \vct{1} \}$.

As a final point, we consider how different pricing schemes affect the purchasing decisions of buyers with robust utility functions. The following plot shows how adjusting prices to reflect uncertainty (in red) results in much larger revenue than if prices were set as though there was no uncertainty (in blue).
\begin{center}
\includegraphics[width=0.6\columnwidth]{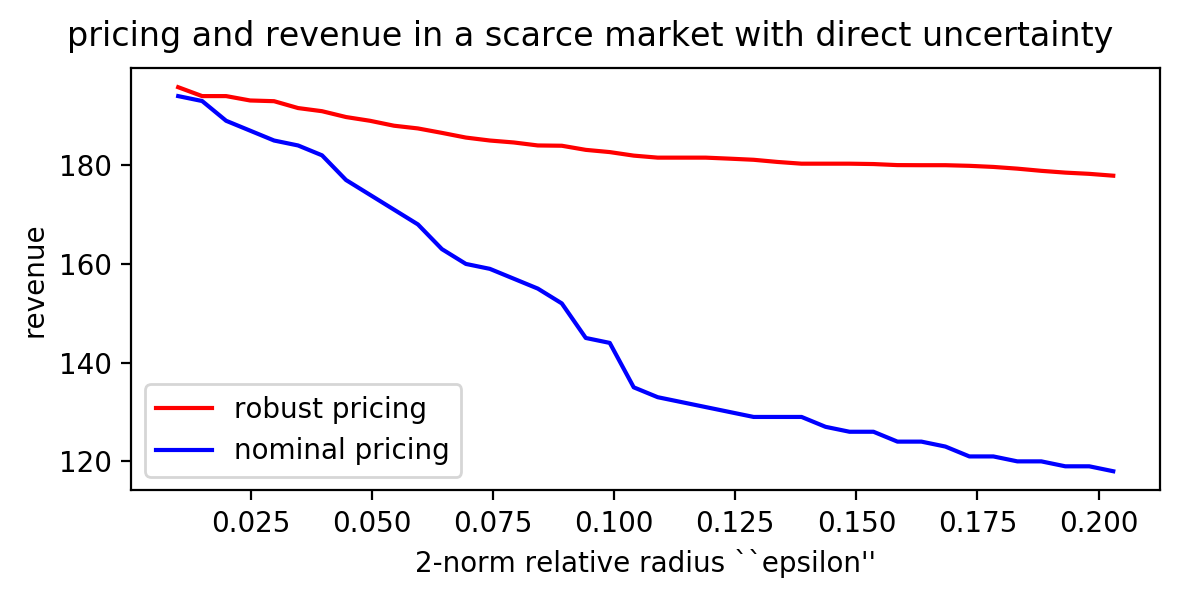}
\end{center}

\section{Acknowledgements}

The majority of this research was conducted while the first and second authors were employed by Facebook, as intern and postdoc, respectively. The first author was supported (in part) by a National Science Foundation Graduate Research Fellowship.

\begin{small}
\bibliographystyle{aaai}
\bibliography{refs}

\begin{thebibliography}{}

\bibitem[\protect\citeauthoryear{Aghassi and
  Bertsimas}{2006}]{aghassi2006robust}
Aghassi, M., and Bertsimas, D.
\newblock 2006.
\newblock Robust game theory.
\newblock {\em Mathematical Programming} 107(1-2):231--273.

\bibitem[\protect\citeauthoryear{Agrawal \bgroup et al\mbox.\egroup
  }{2018}]{cvxpy_rewriting}
Agrawal, A.; Verschueren, R.; Diamond, S.; and Boyd, S.
\newblock 2018.
\newblock A rewriting system for convex optimization problems.
\newblock {\em Journal of Control and Decision} 5(1):42--60.

\bibitem[\protect\citeauthoryear{Albert \bgroup et al\mbox.\egroup
  }{2017}]{albert2017mechanism}
Albert, M.; Conitzer, V.; Lopomo, G.; and Stone, P.
\newblock 2017.
\newblock Mechanism design with correlated valuations: Efficient methods for
  revenue maximization.
\newblock {\em Under Submission at Operations Research}.

\bibitem[\protect\citeauthoryear{Balseiro and Gur}{2017}]{BalseiroGur}
Balseiro, S.~R., and Gur, Y.
\newblock 2017.
\newblock Learning in repeated auctions with budgets: Regret minimization and
  equilibrium.
\newblock In {\em Proceedings of the 2017 ACM Conference on Economics and
  Computation}, EC '17,  609--609.
\newblock New York, NY, USA: ACM.

\bibitem[\protect\citeauthoryear{Balseiro, Besbes, and
  Weintraub}{2015}]{balseiro2015repeated}
Balseiro, S.~R.; Besbes, O.; and Weintraub, G.~Y.
\newblock 2015.
\newblock Repeated auctions with budgets in ad exchanges: Approximations and
  design.
\newblock {\em Management Science} 61(4):864--884.

\bibitem[\protect\citeauthoryear{Balseiro \bgroup et al\mbox.\egroup
  }{2017}]{balseiro2017budget}
Balseiro, S.; Kim, A.; Mahdian, M.; and Mirrokni, V.
\newblock 2017.
\newblock Budget management strategies in repeated auctions.
\newblock In {\em Proceedings of the 26th International Conference on World
  Wide Web},  15--23.
\newblock International World Wide Web Conferences Steering Committee.

\bibitem[\protect\citeauthoryear{Ben-Tal and Nemirovski}{2002}]{ben2002robust}
Ben-Tal, A., and Nemirovski, A.
\newblock 2002.
\newblock Robust optimization--methodology and applications.
\newblock {\em Mathematical Programming} 92(3):453--480.

\bibitem[\protect\citeauthoryear{Ben-Tal, El~Ghaoui, and
  Nemirovski}{2009}]{ben2009robust}
Ben-Tal, A.; El~Ghaoui, L.; and Nemirovski, A.
\newblock 2009.
\newblock {\em Robust optimization}, volume~28.
\newblock Princeton University Press.

\bibitem[\protect\citeauthoryear{Bergemann and
  Morris}{2005}]{bergemann2005robust}
Bergemann, D., and Morris, S.
\newblock 2005.
\newblock Robust mechanism design.
\newblock {\em Econometrica} 73(6):1771--1813.

\bibitem[\protect\citeauthoryear{Bertsimas and Sim}{2004}]{bertsimas2004price}
Bertsimas, D., and Sim, M.
\newblock 2004.
\newblock The price of robustness.
\newblock {\em Operations research} 52(1):35--53.

\bibitem[\protect\citeauthoryear{Bertsimas, Brown, and
  Caramanis}{2011}]{bertsimas2011theory}
Bertsimas, D.; Brown, D.~B.; and Caramanis, C.
\newblock 2011.
\newblock Theory and applications of robust optimization.
\newblock {\em SIAM review} 53(3):464--501.

\bibitem[\protect\citeauthoryear{Birnbaum, Devanur, and
  Xiao}{2011}]{birnbaum2011distributed}
Birnbaum, B.; Devanur, N.~R.; and Xiao, L.
\newblock 2011.
\newblock Distributed algorithms via gradient descent for fisher markets.
\newblock In {\em Proceedings of the 12th ACM conference on Electronic
  commerce},  127--136.
\newblock ACM.

\bibitem[\protect\citeauthoryear{Borgs \bgroup et al\mbox.\egroup
  }{2007}]{borgs2007dynamics}
Borgs, C.; Chayes, J.; Immorlica, N.; Jain, K.; Etesami, O.; and Mahdian, M.
\newblock 2007.
\newblock Dynamics of bid optimization in online advertisement auctions.
\newblock In {\em Proceedings of the 16th international conference on World
  Wide Web}.

\bibitem[\protect\citeauthoryear{Budish and Cantillon}{2012}]{budish2012multi}
Budish, E., and Cantillon, E.
\newblock 2012.
\newblock The multi-unit assignment problem: Theory and evidence from course
  allocation at harvard.
\newblock {\em American Economic Review} 102(5):2237--71.

\bibitem[\protect\citeauthoryear{Chakraborty, Devanur, and
  Karande}{2010}]{chakraborty2010market}
Chakraborty, S.; Devanur, N.~R.; and Karande, C.
\newblock 2010.
\newblock Market equilibrium with transaction costs.
\newblock In {\em International Workshop on Internet and Network Economics},
  496--504.
\newblock Springer.

\bibitem[\protect\citeauthoryear{Chen, Ye, and Zhang}{2007}]{chen2007note}
Chen, L.; Ye, Y.; and Zhang, J.
\newblock 2007.
\newblock A note on equilibrium pricing as convex optimization.
\newblock In {\em International Workshop on Web and Internet Economics},
  7--16.
\newblock Springer.

\bibitem[\protect\citeauthoryear{Cole \bgroup et al\mbox.\egroup
  }{2017}]{cole2017convex}
Cole, R.; Devanur, N.~R.; Gkatzelis, V.; Jain, K.; Mai, T.; Vazirani, V.~V.;
  and Yazdanbod, S.
\newblock 2017.
\newblock Convex program duality, fisher markets, and nash social welfare.
\newblock In {\em 18th ACM Conference on Economics and Computation, EC 2017}.
\newblock Association for Computing Machinery, Inc.

\bibitem[\protect\citeauthoryear{Conitzer \bgroup et al\mbox.\egroup
  }{2018}]{conitzer2018multiplicative}
Conitzer, V.; Kroer, C.; Sodomka, E.; and Stier-Moses, N.~E.
\newblock 2018.
\newblock Multiplicative pacing equilibria in auction markets.
\newblock In {\em International Conference on Web and Internet Economics}.

\bibitem[\protect\citeauthoryear{Conitzer \bgroup et al\mbox.\egroup
  }{2019}]{conitzer2019pacing}
Conitzer, V.; Kroer, C.; Panigrahi, D.; Schrijvers, O.; Sodomka, E.;
  Stier-Moses, N.~E.; and Wilkens, C.
\newblock 2019.
\newblock Pacing equilibrium in first-price auction markets.
\newblock In {\em Proceedings of the 2019 ACM Conference on Economics and
  Computation}.
\newblock ACM.

\bibitem[\protect\citeauthoryear{Dahl and Andersen}{2019}]{dahl2019primal}
Dahl, J., and Andersen, E.~D.
\newblock 2019.
\newblock A primal-dual interior-point algorithm for nonsymmetric
  exponential-cone optimization.

\bibitem[\protect\citeauthoryear{Diamond and Boyd}{2016}]{cvxpy}
Diamond, S., and Boyd, S.
\newblock 2016.
\newblock Cvxpy: A python-embedded modeling language for convex optimization.
\newblock {\em The Journal of Machine Learning Research} 17(1):2909--2913.

\bibitem[\protect\citeauthoryear{Domahidi, Chu, and
  Boyd}{2013}]{domahidi2013ecos}
Domahidi, A.; Chu, E.; and Boyd, S.
\newblock 2013.
\newblock {ECOS}: {A}n {SOCP} solver for embedded systems.
\newblock In {\em European Control Conference (ECC)},  3071--3076.

\bibitem[\protect\citeauthoryear{Eisenberg and
  Gale}{1959}]{eisenberg1959consensus}
Eisenberg, E., and Gale, D.
\newblock 1959.
\newblock Consensus of subjective probabilities: The pari-mutuel method.
\newblock {\em The Annals of Mathematical Statistics} 30(1):165--168.

\bibitem[\protect\citeauthoryear{Eisenberg}{1961}]{eisenberg1961aggregation}
Eisenberg, E.
\newblock 1961.
\newblock Aggregation of utility functions.
\newblock {\em Management Science} 7(4):337--350.

\bibitem[\protect\citeauthoryear{Kroer \bgroup et al\mbox.\egroup
  }{2019}]{kroer2019computing}
Kroer, C.; Peysakhovich, A.; Sodomka, E.; and Stier-Moses, N.~E.
\newblock 2019.
\newblock Computing large market equilibria using abstractions.
\newblock {\em arXiv preprint arXiv:1901.06230}.

\bibitem[\protect\citeauthoryear{Lopomo, Rigotti, and
  Shannon}{2018}]{lopomo2018uncertainty}
Lopomo, G.; Rigotti, L.; and Shannon, C.
\newblock 2018.
\newblock Uncertainty and robustness of surplus extraction.
\newblock {\em arXiv preprint arXiv:1811.01320}.

\bibitem[\protect\citeauthoryear{{MOSEK ApS}}{2019}]{mosek2010mosek}
{MOSEK ApS}.
\newblock 2019.
\newblock The mosek optimization software. version 9.
\newblock {\em Online at http://www.mosek.com}.

\bibitem[\protect\citeauthoryear{Peysakhovich and
  Kroer}{2019}]{peysakhovich2019fair}
Peysakhovich, A., and Kroer, C.
\newblock 2019.
\newblock Fair division without disparate impact.
\newblock {\em arXiv preprint arXiv:1906.02775}.

\bibitem[\protect\citeauthoryear{Peysakhovich, Kroer, and
  Lerer}{2019}]{peysakhovich2019robust}
Peysakhovich, A.; Kroer, C.; and Lerer, A.
\newblock 2019.
\newblock Robust multi-agent counterfactual prediction.
\newblock {\em arXiv preprint arXiv:1904.02235}.

\bibitem[\protect\citeauthoryear{Roth}{2015}]{roth2015gets}
Roth, A.~E.
\newblock 2015.
\newblock {\em Who Gets What?and Why: The New Economics of Matchmaking and
  Market Design}.
\newblock Houghton Mifflin Harcourt.

\end{thebibliography}
\end{small}



\section*{Appendices}

\subsection*{Appendix A: 1-norm experiments}

Here we repeat the experiments from Section \ref{sec:experimental_results}, this time using 1-norm uncertainty rather than 2-norm uncertainty. 
 
We begin with Nash welfare. In the 2-norm experiments we claimed that the robust Nash welfare of a nominal solution is very sensitive to direct uncertainty, while it is relatively stable for buyerside uncertainty.
This is even more true when working with uncertainty derived from the 1-norm.
Consider the left panel of the first figure below: a very small perturbation $\epsilon = 0.01$ saw the robust Nash welfare of the nominal solution drop two orders of magnitude.
By contrast, both nominal and robust Nash welfare of the robust solution declined by only a few percentage points.

\begin{center}
\includegraphics[width=0.6\columnwidth]{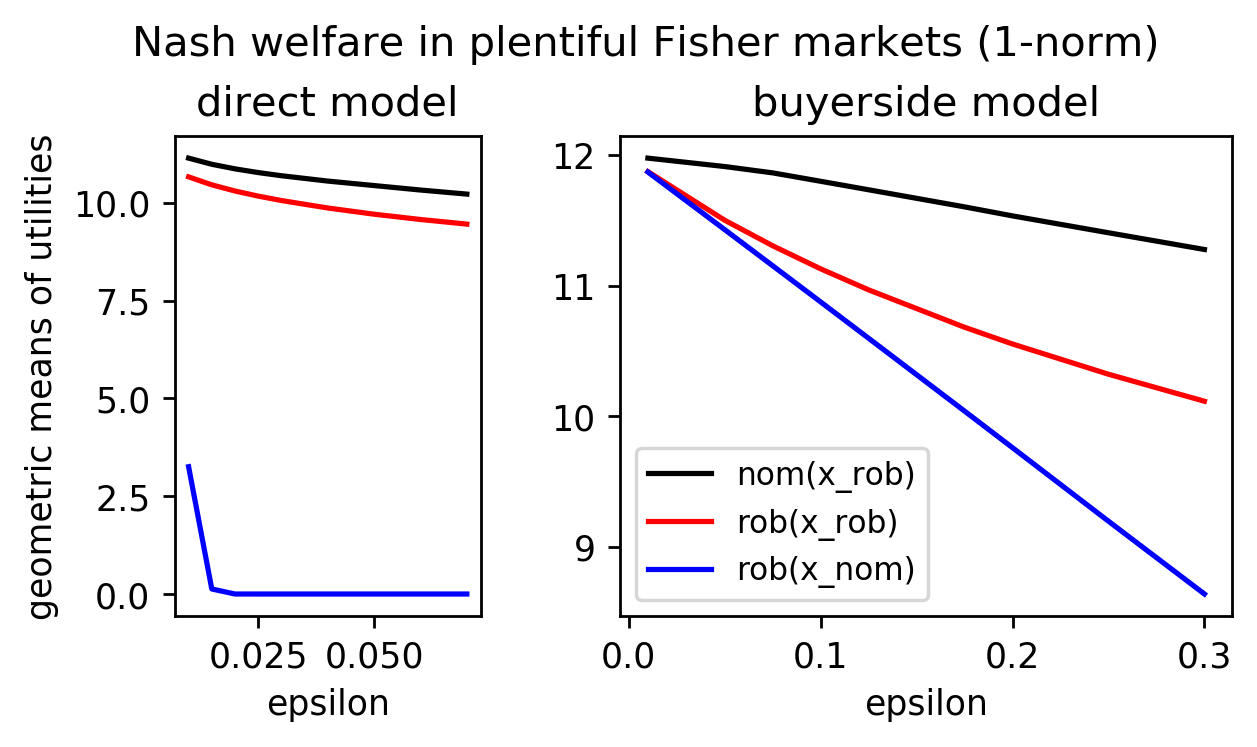}
\end{center}

\begin{center}
\includegraphics[width=0.6\columnwidth]{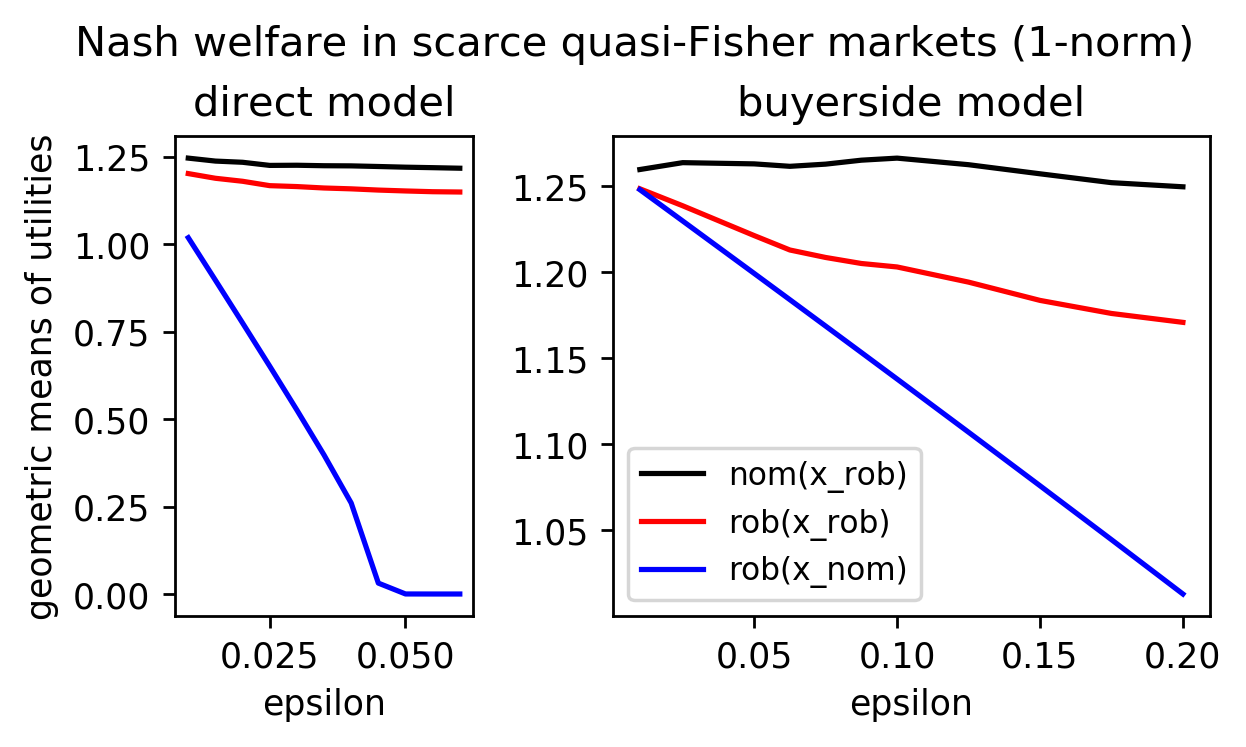}
\end{center}

Next we consider robust envy. 
The following figure shows how the robust solution produces far better outcomes than the nominal solution for the plentiful Fisher market under direct-model $\mathcal{V}_i^{d}(1,0.03)$ uncertainty.
This is not the case for buyerside uncertainty.
Even with a large value of $\epsilon = 0.2$, the robust-envy distributions induced by nominal and robust solutions differ primarily by a small shift in mean.

\begin{center}
\includegraphics[width=0.6\columnwidth]{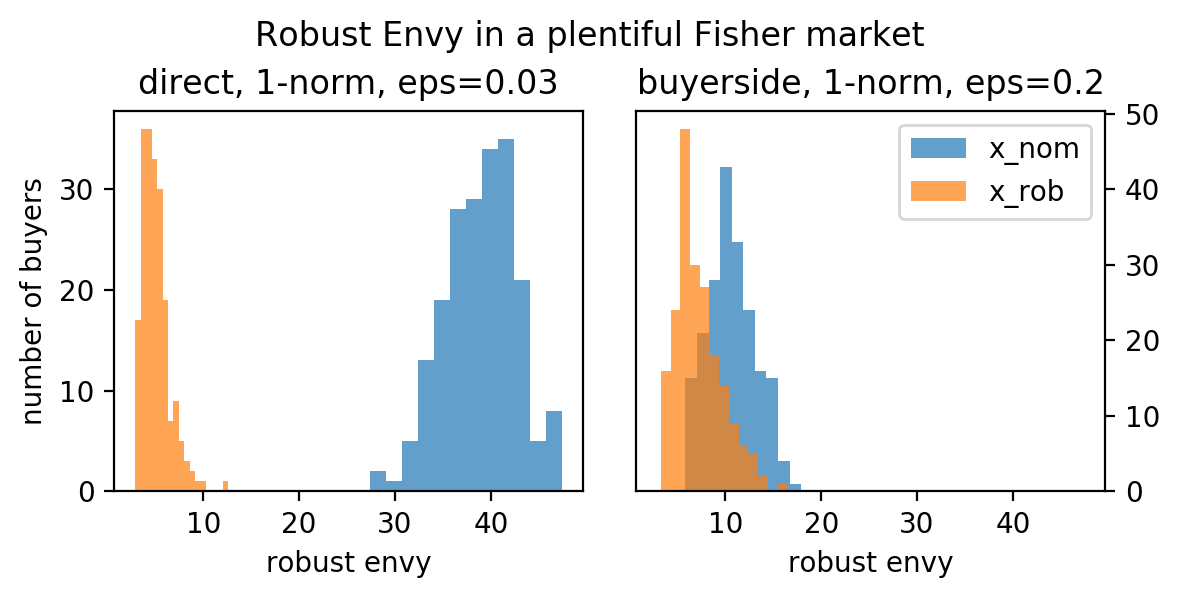}
\end{center}

The robust-envy benefits of using robust solutions persist in our scarce market, although the effects are less dramatic here. In the following figure, note that the distributions of direct-model robust-envy induced by robust and nominal solutions are \textit{disjoint}, when they exhibited overlap in the 2-norm setting.

\begin{center}
\includegraphics[width=0.6\columnwidth]{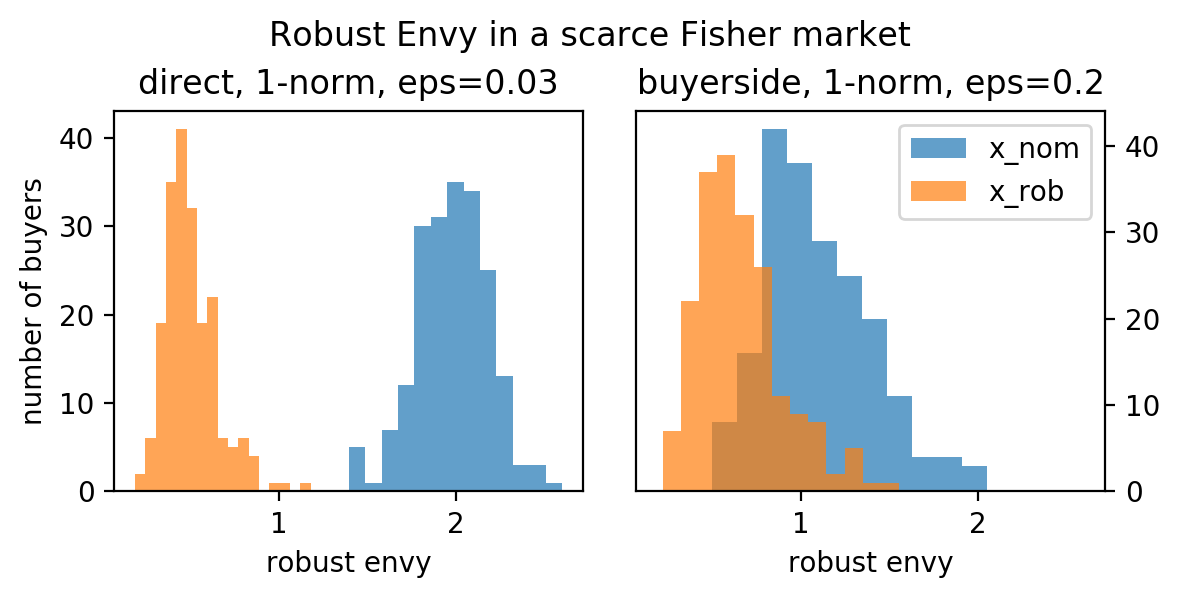}
\end{center}

Now we turn to how uncertainty affects equilibrium prices.
These experiments only consider quasi-Fisher markets, since buyers in these settings have a choice of whether or not to participate in the market.

The figure below should be read as follows: the price of every good $j$ is plotted as the uncertainty radius $\epsilon$ ranges from 0 to 0.5.
These lines have reduced opacity, so that areas of higher price density can be easily discerned. The minimum, maximum, and mean prices are traced by solid black lines in the upper, lower, and middle portions of the plots respectively. 

\begin{center}
\includegraphics[width=0.6\columnwidth]{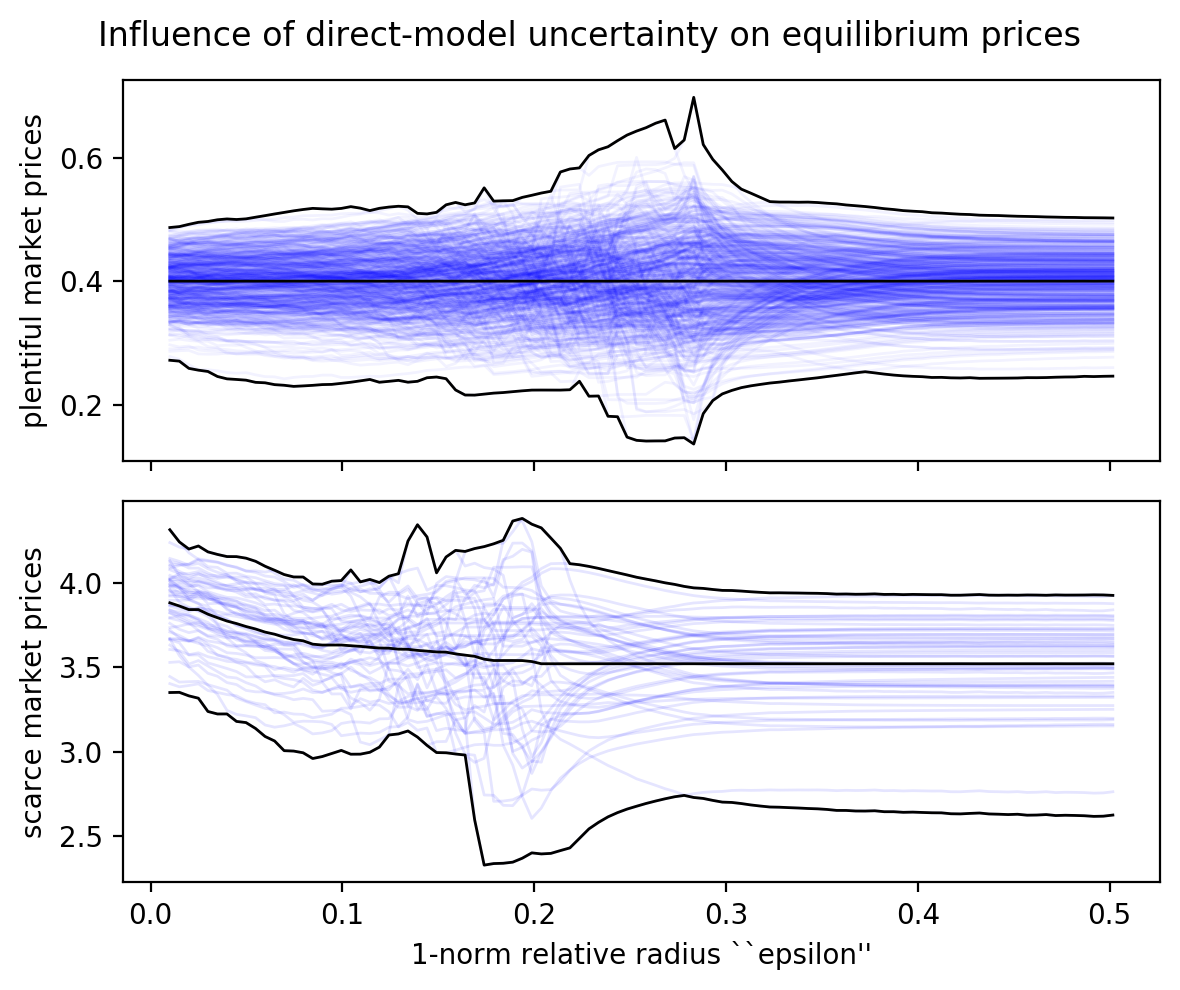}
\end{center}

There are important similarities between the plots above and the analogous plots in 2-norm experiments. First, we note that prices do not evolve monotonically as $\epsilon$ increases, and second, the ordering on goods induced by sorting $p_j$ can change as a result of changes in uncertainty. There is also a phase-transition from jagged and oscillatory price curves, to smooth and stable price curves.
The behavior of the smooth regions under this 1-norm uncertainty is qualitatively different from the 2-norm case, in that prices do not seem to converge to a common value even as $\epsilon$ approaches a large value of 0.5.
We note that when replicating these plots for extremely large values of $\epsilon$ (e.g. $\epsilon = 2$ to $\epsilon = 3$) it is possible to make prices converge to a common value, however such levels of uncertainty are entirely unreasonable in practice.

The last component of our experiments concerns how different pricing schemes would affect revenue for agents which seek to maximize their robust utility. The plot below qualitatively matches what we saw in the 2-norm case: using nominal prices for buyers with robust utility functions produces much less revenue than if robust prices were used. Reasonable perturbations of 10\% relative error in the valuations can result in revenue dropping by over 20\%.

\begin{center}
\includegraphics[width=0.6\columnwidth]{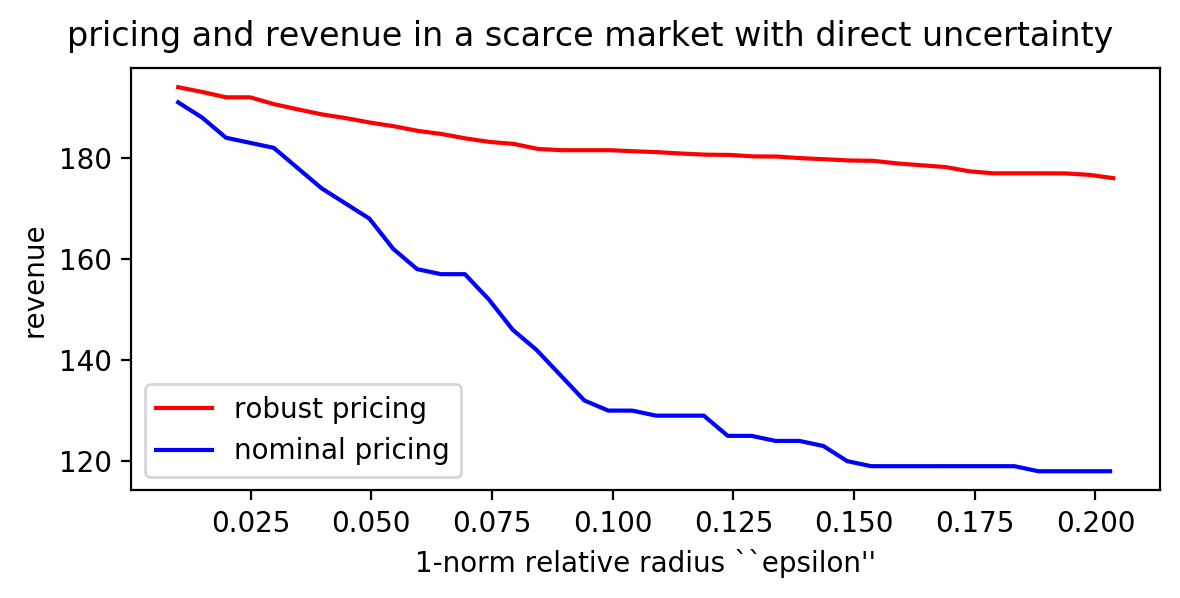}
\end{center}


\subsection*{Appendix B: deferred proofs}

This appendix provides the omitted proofs for claims made throughout the article. In particular, we provide proofs of Propositions \ref{prop:validdual}, \ref{prop:avg_regret}, and \ref{prop:joint_outer_approx}. We also provide a direct proof of Proposition \ref{prop:Fisher}, which was previously established indirectly through appeal to \citet{eisenberg1961aggregation}.

\begin{customprop}{3}\label{prop:validdual}
Let $u_i(\vct{z}) = \min\{\vct{z} \cdot \vct{v} \,\mid\, \vct{v} \in \mathcal{V}_i\}$. Then the optimal objective value of \eqref{eq:EG} is no larger than the optimal objective value of \eqref{eq:EG_dual}.
\end{customprop}

\begin{proof}
Let $\mathrm{Opt}$ denote the optimal objective value of \eqref{eq:EG}.
We will construct \eqref{eq:EG_dual} by considering \eqref{eq:EG} as minimizing the negative objective $\sum_{i=1}^n \mathrm{Q}\, r_i - b_i \log(t_i)$.
We will undo this negation as the final step of the proof.

Begin by introducing dual variables $\vct{\beta} \in \R^n_+$ for the utility hypograph constraints, and $\vct{p} \in \R^m_+$ for the item supply constraints.
We also introduce variables $\vct{v}_i \in \mathcal{V}_i$ to simplify our Lagrangian:
\begin{align*}
\mathcal{L}(\mtx{X},\vct{t},\vct{r},\vct{\beta},\vct{p},\mtx{V}) = & 
     \textstyle\sum_{i=1}^n\{  \mathrm{Q}\, r_i - b_i \log(t_i) \} \\
     & + \textstyle\sum_{i=1}^n\{ \beta_i(t_i - \vct{v}_i \cdot \vct{x}_i - \mathrm{Q}\, r_i) \} \\
     & + \textstyle\sum_{i=1}^n\{ \vct{x}_i \cdot \vct{p} \} - \vct{1} \cdot \vct{p}.
\end{align*}
This Lagrangian is useful to us, because we have the following minimax characterization of $\mathrm{Opt}$:
\begin{align}
-\mathrm{Opt} =& \min_{\substack{\mtx{X} \geq \mtx{0} \\ \vct{r} \geq \vct{0} \\ \vct{t} \in \R^n }} \max_{\substack{\vct{p} \geq \vct{0} \\ \vct{\beta} \geq \vct{0} \\ \vct{v}_i \in \mathcal{V}_i}} \mathcal{L}(\mtx{X},\vct{t},\vct{r},\vct{\beta},\vct{p},\mtx{V}) \label{eq:minimax_eg_partialdual} \\
 \geq & \max_{\substack{\vct{p} \geq \vct{0} \\ \vct{\beta} \geq \vct{0} \\ \vct{v}_i \in \mathcal{V}_i}} \underbrace{\min_{\substack{\mtx{X} \geq \mtx{0} \\ \vct{r} \geq \vct{0} \\ \vct{t} \in \R^n }} \mathcal{L}(\mtx{X},\vct{t},\vct{r},\vct{\beta},\vct{p},\mtx{V})}_{ =: F(\vct{\beta},\vct{p},\mtx{V}) }. \label{eq:maximin_eg_partialdual}
\end{align}
In view of the minimax relationship \eqref{eq:minimax_eg_partialdual}-\eqref{eq:maximin_eg_partialdual}, the claim of weak duality for the pair \eqref{eq:EG}-\eqref{eq:EG_dual} reduces to simplifying the dual function $F(\vct{\beta},\vct{p},\mtx{V})$ to a form which matches \eqref{eq:EG_dual}.

Next we group terms by primal variables. Define
\begin{itemize}
    \item $A(\mtx{X},\vct{\beta},\vct{p},\mtx{V}) = \sum_{i=1}^n \vct{x}_i \cdot \vct{p} - \beta_i \vct{v}_i \cdot \vct{x}_i$,
    \item $B(\vct{t},\vct{\beta}) = \sum_{i=1}^n \beta_i t_i - b_i \log(t_i)$, and
    \item $C(\vct{r},\vct{\beta}) = \sum_{i=1}^n r_i - \beta_i r_i$
\end{itemize}
so that the dual function $F(\vct{\beta},\vct{p},\mtx{V})$ is given by
\begin{equation*}
\min_{\substack{\mtx{X} \geq \mtx{0} \\ \vct{r} \geq \vct{0} \\ \vct{t} \in \R^n }} A(\mtx{X},\vct{\beta},\vct{p},\mtx{V}) + B(\vct{t},\vct{\beta}) + \mathrm{Q}\, C(\vct{r},\vct{\beta}) - \vct{1} \cdot \vct{p}.
\end{equation*}
The minimization over $\mtx{X}$, $\vct{r}$, and $\vct{t}$ can be performed independently from one another, and yields
\begin{align}
& \min_{\mtx{X} \geq \mtx{0}}A(\mtx{X},\vct{\beta},\vct{p},\mtx{V}) = \begin{cases}
    0 & \text{if } \vct{p} - \beta_i \vct{v}_i \geq \vct{0} ~ \forall i \\
    -\infty &\text{ otherwise}
    \end{cases}, \label{eq:xij_and_bilinear_compslack}  \\
& \min_{\vct{t} \in \R^n} B(\vct{t},\vct{\beta}) = \sum_{i=1}^n b_i + b_i \log(\beta_i / b_i),  \label{eq:Bterm_validdual}
\end{align}
and
\begin{equation}
\min_{\vct{r} \geq \vct{0}} C(\vct{r},\vct{\beta}) = \begin{cases} 0 & \text{ if } 1 - \vct{\beta} \geq \vct{0} \\
 -\infty &\text{ otherwise}
 \end{cases}.\label{eq:Cterm_validdual}
\end{equation}

We have thus shown that
\[
\mathrm{Opt} \leq \min_{\substack{\vct{p} \geq \vct{0} \\ \vct{\beta} \geq \vct{0} \\ \vct{v}_i \in \mathcal{V}_i}} -F(\vct{\beta},\vct{p},\mtx{V}),
\]
and by simplifying the dual function with \eqref{eq:xij_and_bilinear_compslack}-\eqref{eq:Cterm_validdual}, this is equivalent to $\mathrm{Opt}$ being a lower bound on the objective of problem \eqref{eq:EG_dual}. 
\end{proof}

It is also possible to prove a \textit{strong duality} result between \eqref{eq:EG}-\eqref{eq:EG_dual}.
From a technical perspective, strong duality would require that $\emptyset \neq \mathcal{V}_i \subset \R^m_+$ are compact convex sets, and $u_i(\vct{1}) > 0$.
Indeed, these are the assumptions stated in Section \ref{sec:market_defs}.
A proof of strong duality would proceed by using the facts that $\vct{z} \mapsto -u_i(-\vct{z})$ is the support function of the convex set $\mathcal{V}_i$, and that the Fenchel conjugate of a convex set's support function is the set's indicator function.

As a mechanical detail, it would be important that to introduce dual variables to constraints $x_{ij} \geq 0$ when using the Fenchel duality approach.
Dual variables for such constraints play an important role in our supplemental proof of Proposition \ref{prop:Fisher}, given below.\footnote{We call this proof \textit{supplemental}, because the results of Eisenberg (1961) and Chen et al. (2007) suffice for indirect proofs.}

\begin{customprop}{1}
A solution to \eqref{eq:EG} with $\mathrm{Q} = 0$ produces an equilibrium for Fisher markets under robust utilities
\[u_i(\vct{z}) =\min\{ \vct{z} \cdot \vct{v} \,\mid\, \vct{v} \in \mathcal{V}_i \}.
\] 
The solution produces optimal prices ${\vct{p}}^*$ and allocations $\mtx{X}^*$ with the following properties:
\begin{enumerate}
    \item Optimal allocations $\vct{x}_i^*$ are in the robust demand set of buyer $i$ for every $i$, i.e. 
    \[\vct{x}_i^* \in \argmax \left\{\min_{\vct{v} \in \mathcal{V}_i} \vct{v} \cdot \vct{z} \; | \; {\vct{p}}^* \cdot \vct{z} \leq b_i, \; \vct{z} \geq \vct{0}  \right\}.
    \]
    \item Every item $j$ with a positive price $p^*_j > 0$ clears the market, i.e $\sum_i x_{ij}^* = 1$.
    \item For every buyer $i$, there exists a $\vct{v}_i^* \in \argmin_{\vct{v} \in \mathcal{V}_i} \vct{v} \cdot \vct{x}_i$ for which allocated items have the same bang per buck:
     \[
    \text{if } x_{ij}^*, x_{ik}^*>0 \text { then } \frac{v_{ij}^{*}}{p^*_j} = \frac{v_{ik}^{*}}{p^*_k}.
    \]
\end{enumerate}
\end{customprop}
\begin{proof}[Supplemental proof]
Assume strong duality between the primal-dual pair \eqref{eq:EG}-\eqref{eq:EG_dual}.
\begin{itemize}
    \item The complementary slackness condition 
    \[
        p^*_j \cdot \left(\sum_i x_{ij}^* - 1 \right)=0
    \]
    implies the market clearing condition.
 \item Let ${\vct{p}}^*$, $\vct{\beta}^*$, and $\{\vct{v}_i^*\}_{i=1}^n$ be an optimal dual solution, and let $\vct{x}_i^*$ be an optimal primal solution for the primal-dual pair \eqref{eq:EG}-\eqref{eq:EG_dual}.
 By considering Equation \ref{eq:xij_and_bilinear_compslack} in our proof of Proposition \ref{prop:validdual}, there is a complementary-slackness relationship between primal bound constraints $0 \leq x_{ij}$ and dual bilinear constraints $p_j - \beta_i v_{ij} \geq 0$.
 Thus whenever  $x_{ij}^* >0$, we have ${p}_j^* - \beta_i^* v_{ij}^*=0$, and hence the ratio $v_{ij}^* / p_j^* = \beta_i^*$ is the same for all items $j$ where $x_{ij}^* > 0$. 
    
    Next we must show that $\vct{v}_i^*$ belongs to $\argmin_{\vct{v} \in \mathcal{V}_i} \vct{v} \cdot \vct{x}_i$.
    This follows by considering the dual Lagrangian
    \begin{align*}
      (\vct{p}^*,\vct{\beta}^*,\mtx{V},\mtx{X}^*) \mapsto 
        &  ~~ {\vct{p}}^* \cdot \vct{1} - \textstyle\sum_{i=1}^n\{ b_i + b_i \log(\beta_i^*/b_i) \}  \\
        & ~~ - \textstyle\sum_{i=1}^n \{ \vct{x}_i^* \cdot ({\vct{p}}^* - \beta_i^* \vct{v}_i) \}
    \end{align*}
    which (per first order necessary conditions) must be minimized over $\vct{v}_i \in \mathcal{V}_i$. Since $\beta_i^* > 0$, the minimizers $\vct{v}_i^*$ will always belong to $\argmin_{\vct{v} \in \mathcal{V}_i} \vct{v} \cdot \vct{x}_i^*$.
    This implies the third assertion.

    \item For the first assertion, we note that the first order KKT condition implied by the vanishing of the dual Lagrangian derivative with respect to $\beta_i$ implies that ${\vct{p}}^* \cdot \vct{x}_i^* = b_i$. Now consider the dual optimization problem faced by buyer $i$:
  \begin{align*}
\min &~ \textstyle \lambda_i b_i  \\
\text{ s.t.}
    & ~~ \vct{v}_i \leq \lambda_i {\vct{p}}^* \nonumber \\
    & \vct{v}_i \in \mathcal{V}_i 
\end{align*}
We note that $\lambda_i = {\beta_i^*}^{-1}, \vct{v}_i^*$ are feasible with respect to this dual problem, and thus the quantity $\frac{b_i}{\beta_i^*}$ constitutes an upper bound to the utility attainable by buyer $i$. But this utility is attained at the allocation $\vct{x}_i^*$ since
$$
\min_{\vct{v}_i \in \mathcal{V}_i} \vct{v}_i \cdot \vct{x}_i^* = \vct{v}_i^* \cdot \vct{x}_i^* = \sum_j \frac{p_j}{\beta_i^*}  x_{ij}^* = \frac{b_i}{\beta_i*}.
$$
Hence $\vct{x}_i^*$ is in the demand set of buyer $i$.
\end{itemize}
\end{proof}

The story for the quasi-Fisher case (Proposition 2) is similar.
One is free to rely on prior work by Chen et al. to be certain that the primal allocations and dual variables to supply constraints comprise an equilibrium for $u_i$ of the form we consider.
Just as well, one can prove the proposition by considering optimality conditions of the primal-dual pair \eqref{eq:EG}-\eqref{eq:EG_dual}.
Since the Fisher market illustrated the latter approach, we do not dwell on this for the quasi-Fisher case.
Instead, we turn to proving the second robust-regret bound from Section 3.3.

\begin{customprop}{5}
Suppose $|v_{ij} - v_{ij}^{\prime}| \leq R$ for all $\vct{v}_{i}, \vct{v}_{i}^{\prime} \in \mathcal{V}_i$, $i=1, \ldots, n$. Then the following holds:
$$
\frac{1}{n}\sum_{i=1}^n \text{Robust-Regret}_i \leq \frac{2Rm}{n}.
$$
\end{customprop}
\begin{proof}
Suppose $\vct{x}_i, \vct{x}_i^{\prime}$ are two feasible allocations, and $\vct{v}_i, \vct{v}_i^{\prime}$ for $i=1, \ldots, n$ are utility parameters in the uncertainty set $\mathcal{V}_i$ such that 
$$\vct{v}_i \cdot \vct{x}_i \geq \vct{v}_i \cdot \vct{x}_i^{\prime}.$$
Setting $\Delta \vct{v}_i = \vct{v}_i - \vct{v}_i^{\prime}$ it follows that
\begin{align*}
    \vct{v}_i^{\prime} \cdot \vct{x}_i &=  \vct{v}_i \cdot \vct{x}_i - \Delta \vct{v}_i \cdot \vct{x}_i \\
    & \geq \vct{v}_i \cdot \vct{x}_i^{\prime} - \Delta \vct{v}_i \cdot \vct{x}_i \\
    & = \vct{v}_i^{\prime} \cdot \vct{x}_i^{\prime} + \Delta \vct{v}_i \cdot \left( \vct{x}_i^{\prime} - \vct{x}_i \right).
\end{align*}
Hence,
\begin{align*}
\frac{1}{n}\sum_{i=1}^n  \left(\vct{v}_i^{\prime} \cdot \vct{x}_i -  \vct{v}_i^{\prime} \cdot \vct{x}_i^{\prime} \right) & \geq \frac{1}{n}\sum_{i=1}^n \Delta \vct{v}_i \cdot \left( \vct{x}_i^{\prime} - \vct{x}_i \right) \\
& = \frac{1}{n}\sum_{i=1}^n \sum_{j=1}^m \Delta v_{ij}  \left( x_{ij}^{\prime} - x_{ij} \right) \\
& \geq -\frac{2}{n}\sum_{i=1}^n \sum_{j=1}^m |\Delta v_{ij}|  x_{ij} \\
& \geq -\frac{2}{n}\sum_{i=1}^n \sum_{j=1}^m R  x_{ij} \\
& = -2Rm / n.
\end{align*}
Setting $\vct{v}_i$ to be the realization with respect to which the equilibrium allocation $\vct{x}_i$ for $i=1, \ldots, n$ is calculated, and $\vct{v}_i^{\prime}$, $\vct{x}_i^{\prime}$ to be the allocation at which the worst-case robust-regret is realized, we obtain the required bound on robust regret.
\end{proof}

Finally, we turn to proving Proposition \ref{prop:joint_outer_approx}. The proof will largely reduce to the following simple lemma. 

\begin{lemma} \label{lemma:joint0}
Set $S = \left\{ \vct{\theta} \mtx{\Phi}^{\intercal} \; : \; \|\vct{\theta} \|_2 \leq \epsilon_1, \; \| \mtx{\Phi}\|_F \leq \epsilon_2 \right\}$ and $T = \left\{ \vct{v}: \|\vct{v}\|_2 \leq \epsilon_1 \epsilon_2 \right\}$. We claim that $S=T$.
\end{lemma}
\begin{proof}
Note that the origin is contained in both sets. Consider a point $\vct{x} \neq \vct{0}$, and
suppose $\vct{x} \in T$. By choosing $\vct{\theta} = \epsilon_1[1,0,\ldots,0]$, and $\mtx{\Phi}^\intercal$ to be the matrix where the first row is $\epsilon_2 \vct{x}^\intercal  / \|\vct{x}\|_2$ and all remaining rows zero, we see that $\vct{x} = \vct{\theta}\mtx{\Phi}^\intercal$, and hence $\vct{x} \in S$.

Suppose $\vct{x} \doteq \vct{\theta}\mtx{\Phi}^\intercal \in S$. Viewing $\mtx{\Phi}^\intercal$ as a linear operator acting on $\vct{\theta}$ and noting the relationship between the operator and Frobenius norms $\|\mtx{\Phi}^\intercal \| \leq \|\mtx{\Phi}^\intercal \|_F$, we have that
\begin{align*}
    \|\vct{x}\|_2 &\leq \| \vct{\theta} \|_2 \| \mtx{\Phi}^\intercal \| \\
    & \leq  \| \vct{\theta} \|_2 \| \mtx{\Phi}^\intercal \|_F \\
    & \leq \epsilon_1 \epsilon_2,
\end{align*}
and so $\vct{x} \in  T$.
\end{proof}

\begin{customprop}{6}
Let ${\mathcal{V}}^{out}_i(\epsilon_1, \epsilon_2)$ denote the set
\begin{align*}
\{ \vct{v}  \,:\,
&  \vct{v} = \vct{\mu} + \vct{\delta} \mtx{\widehat{\Phi}}^\intercal + \vct{\hat{\theta}}_i \mtx{\Delta}^\intercal +  \vct{\hat{\theta}}_i \mtx{\widehat{\Phi}}^\intercal, \\
& \vct{v} \geq 0, ~ \vct{v} \cdot \vct{1} = \vct{\hat{v}}_i \cdot \vct{1},~  \| \vct{\delta} \|_2 \leq \epsilon_1 \| \vct{\hat{\theta}}_i \|_2,   \\
& \|\vct{\mu} \|_2 \leq \epsilon_1 \epsilon_2  \| \vct{\hat{\theta}}_i\|_2 \|\mtx{\widehat{\Phi}} \|_F , \; \| \mtx{\Delta} \|_F \leq \epsilon_2 \| \mtx{\widehat{\Phi}} \|_F  \}.
\end{align*}
 Then
 \[
 \conv\left(\mathcal{V}^{J}_i(\epsilon_1, \epsilon_2)\right) \subseteq \mathcal{V}^{out}_i(\epsilon_1, \epsilon_2).
\]
\end{customprop}
\begin{proof}
Note that every element in $\mathcal{V}^{J}_i(\epsilon_1, \epsilon_2)$ may be expressed as $\left(\vct{\hat{\theta}_i} + \vct{\delta}\right)\left(\mtx{\widehat{\Phi}}^\intercal + \mtx{\Delta}^\intercal\right)$, for $\|\vct{\delta}\|_2 \leq \epsilon_1 \|\vct{\hat{\theta}_i}\|_2$, and $\|\mtx{\Delta}\|_F \leq \epsilon_2 \|\mtx{\widehat{\Phi}}\|_F.$ Using Lemma \ref{lemma:joint0}, the result follows immediately.
\end{proof}

\end{document}